\documentclass{article}
\usepackage{geometry}
\usepackage[utf8]{inputenc}

\usepackage[square, sort, numbers]{natbib} 

\usepackage{amsmath}
\usepackage{amssymb}
\usepackage{amsfonts}
\usepackage{amsthm}
\usepackage{epsfig, graphicx, import, color}
\graphicspath{ {./figures/} }
\usepackage{subcaption}
\usepackage[english]{babel}
\usepackage[autostyle]{csquotes}
\usepackage{mathtools}
\usepackage{makecell}
\usepackage{dutchcal}

\numberwithin{equation}{section}
\newtheorem{prop}{Proposition}
\newtheorem*{prop*}{Proposition}

\newtheorem*{lemma*}{Lemma}
\newtheorem{corollary}{Corollary}
\newtheorem*{corollary*}{Corollary}

\theoremstyle{remark}
\newtheorem{remark}{Remark}
\newtheorem*{remark*}{Remark}

\usepackage[title]{appendix}

\usepackage{enumerate}

\title{A Poisson map from kinetic theory to hydrodynamics with non-constant entropy}

\author{Ching Lok Chong}
\date{14 Jan 2022 \\\quad\\OCIAM, Mathematical Institute,
University of Oxford,
Andrew Wiles Building,
\\ Radcliffe Observatory Quarter,
Woodstock Road,
Oxford,
OX2 6GG,
UK}

\begin{document}

\maketitle

\begin{abstract}
Kinetic theory describes a dilute monatomic gas using a distribution function $f(q,p,t)$, the expected phase-space density of particles at a given position $q$ with a given momentum $p$. The distribution function evolves according to the collisionless Boltzmann equation in the high Knudsen number limit. Fluid dynamics provides an alternative description of the gas using macroscopic hydrodynamic variables that are functions of position and time only. The mass, momentum and entropy densities of the gas evolve according to the compressible Euler equations in the limit of vanishing viscosity and thermal diffusivity.
Both systems can be formulated as noncanonical Hamiltonian systems. Each configuration space is an infinite-dimensional Poisson manifold, and the dynamics is the flow generated by a Hamiltonian functional via a Poisson bracket.
We construct a map $\mathcal{J}_1$ from the space of distribution functions to the space of hydrodynamic variables that respects the Poisson brackets on the two spaces. This map is therefore a Poisson map. It maps the $p$-integral of the Boltzmann entropy $f\log f$ to the hydrodynamic entropy density. This map belongs to a family of Poisson maps to spaces that include generalised entropy densities as additional hydrodynamic variables. The whole family can be generated from the Taylor expansion of a further Poisson map that depends on a formal parameter.
If the kinetic-theory Hamiltonian factors through the Poisson map $\mathcal{J}_1$, an exact reduction of kinetic theory to fluid dynamics is possible. However, this is not the case. 
Nonetheless, by ignoring the contribution to the Hamiltonian from the entropy of the distribution function relative to its local Maxwellian, a distribution function defined by the $p$-moments $\int d^np \ (1,p,|p|^2) f$, we construct an approximate Hamiltonian that factors through the map. The resulting reduced Hamiltonian, which depends on the hydrodynamic variables only, generates the compressible Euler equations. We can thus derive the compressible Euler equations as a Hamiltonian approximation to kinetic theory.
We also give an analogous Hamiltonian derivation of the compressible Euler--Poisson equations with non-constant entropy, starting from the Vlasov--Poisson equation.
\end{abstract}

%\tableofcontents

\section{Introduction}\label{sec:intro}

There are different levels of descriptions of a dilute monatomic gas. 
In the $1$-particle kinetic description, the gas is described using a \emph{distribution function} $f(q,p,t)$, which gives the expected phase space density of particles at position $q$ carrying momentum $p$. An uncharged monatomic gas in the Boltzmann--Grad limit evolves according to the \emph{Boltzmann equation} \cite{Cercignani88book, BodineauGallagherSaintRaymond17}:
\begin{align}\label{Boltzmann1}
\frac{\partial f}{\partial t} + p \cdot \frac{\partial f}{\partial q} = \frac{1}{\mathrm{Kn}}C[f],
\end{align}
where $C[f]$ is a quadratic integral operator known as the \emph{Boltzmann collision operator}, and $\mathrm{Kn}$ is the Knudsen number. In the limit $\mathrm{Kn} \rightarrow \infty$, (\ref{Boltzmann1}) becomes the \emph{collisionless Boltzmann equation}. 
A dilute gas of charged particles interacting weakly under the electrostatic force can also be described by $1$-particle kinetic theory. In the Vlasov limit, the distribution function $f$ evolves according to \emph{Vlasov--Poisson equation} \cite{BraunHepp77, Jabin14, Kardar07book}:
\begin{align}\label{Vlasov Poisson}
\frac{\partial f}{\partial t} + p \cdot \frac{\partial f}{\partial q} - \frac{\partial}{\partial q} \left( \int d^3 q^\prime d^3 p^\prime \ \frac{e^2 f(q^\prime,p^\prime,t)}{4\pi \lvert q - q^\prime \rvert} \right) \cdot \frac{\partial f(q,p,t)}{\partial p} = 0,
\end{align}
where $e$ is the suitably rescaled charge of the particles.

An alternative theory that describes the gas is \emph{fluid dynamics}, where we work with macroscopic \emph{hydrodynamic variables} that are functions of position and time only. It is convenient to take the \emph{momentum density} $m(q,t)$, \emph{mass density} $\rho(q,t)$, and the \emph{entropy density} $s(q,t)$ as the hydrodynamic variables. In the limit of vanishing viscosity and thermal diffusivity, the time evolution of the hydrodynamic variables $(m,\rho,s)$ are governed by the \emph{compressible Euler equations}
\begin{align}
\frac{\partial \rho}{\partial t} + \frac{\partial }{\partial q} \cdot m & = 0, \\
\frac{\partial m}{\partial t} + \frac{\partial}{\partial q} \cdot \left( \frac{m \ m}{\rho} \right) & = - \frac{\partial P(\rho,s)}{\partial q}, \\
\frac{\partial s}{\partial t} + \frac{\partial }{\partial q} \cdot \left(\frac{m}{\rho} s\right) & = 0,
\end{align}
where $P(\rho,s)$ is the pressure as determined by the equation of state.

Both $1$-particle kinetic theory and ideal fluid dynamics are \emph{noncanonical Hamiltonian systems} \cite{Morrison98, Salmon88, Holm85}. The time evolution of any functional $F[\cdot]$ of the dynamical variables in each theory is given by
\begin{align}
\frac{dF}{dt} = \left\{ F, H\right\}_{NC},
\end{align}
where $H$ is the \emph{Hamiltonian functional}, and $\{\cdot,\cdot\}_{NC}$ is an \emph{abstract Poisson bracket} or \emph{noncanonical Poisson bracket}. The respective configuration spaces for $1$-particle kinetic theory and fluid dynamics can be treated as infinite-dimensional \emph{Poisson manifolds}, and the time evolution on each space is the flow generated by a Hamiltonian functional $H$ \cite{Morrison98, Marsden83, Marsden84a, Holm85}.

%The passage along the hierarchy of kinetic descriptions, from $N$-body to Liouville to BBGKY to collisionless $1$-particle kinetic theory, is manifestly Hamiltonian \cite{Pavelka15,MarsdenMorrisonWeinstein84}. It is then natural to ask whether the hydrodynamic description can be derived from the $1$-particle kinetic description in a way that respects the Hamiltonian structure.
It is natural to ask whether the hydrodynamic description can be derived from the $1$-particle kinetic description in a manifestly Hamiltonian way.
In this paper, we seek a \emph{Poisson map} $\mathcal{J}_1$ taking $f$ to the hydrodynamic variables $(m[f],\rho[f],s[f])$ that respects the Poisson brackets on the two spaces. 
Given the Poisson map $\mathcal{J}_1$, suppose that the kinetic-theory Hamiltonian $H_{KT}$ factors through the Poisson map $\mathcal{J}_1$, which means that we have another functional $H_{red}$ of the hydrodynamic variables that satisfies $H_{KT} = H_{red} \circ \mathcal{J}_1$. The time evolution of the image $\mathcal{J}_1[f] = (m[f],\rho[f],s[f])$ under the kinetic-theory Hamiltonian is then exactly reproduced by the Hamiltonian system on the hydrodynamic variables generated by $H_{red}$.
We pursue this \emph{reduction of noncanonical Hamiltonian systems} in this paper (\textbf{\ref{subsec:myreduction}}) \cite{Marsden84a, Marsden84b, Marsden13book, Marsden83, Guillemin80}. There are also other ways to obtain reduced Hamiltonian systems from larger ones, which we discuss in section \textbf{\ref{subsec:notmyreduction}} \cite{Meyer73, MontgomeryMarsdenRatiu84, MarsdenWeinstein74, MarsdenRatiu86, Cannas01book, Marsden83, PerinChandreMorrisonTassi15, Tassi15, Tassi16, Tassi17, KaltsasThroumoulopoulosMorrison17,ChesnokovPavlov12,Pavlov14,Yu2000}.

We construct the Poisson map $\mathcal{J}_1$ in section \textbf{\ref{sec:Poissonmap}}. The map $\mathcal{J}_1$ maps the $p$-integral of the \emph{Boltzmann entropy} $f\log f$ onto the hydrodynamic entropy density. It belongs to a family of Poisson maps $\mathcal{J}_A$ that include \emph{generalised entropies} $s_a[f] =\int d^n p \ f(\log f)^a$ for $a=0,1,2,\ldots,A$ as additional hydrodynamic variables (section \textbf{\ref{subsec:JA}}). These Poisson maps are non-linear generalisations of that obtained in \cite{Guillemin80, Marsden83}, where only $m[f]$ and $\rho[f]$ are considered. The non-linearities appearing in these maps are closely related to the eigenfunctions of the Euler differential operator $x(d/dx)$ (section \textbf{\ref{subsec:JA}}). Furthermore, the $\mathcal{J}_A$ arises from the truncated Taylor expansion of a Poisson map $\mathsf{J}$ that depends on a formal parameter $\xi$. This larger map $\mathsf{J}$ has the \emph{Tsallis entropy} $\rho_{\xi}[f] = \int d^n p f^{1+\xi}$ \cite{Tsallis88} as a hydrodynamic variable (section \textbf{\ref{subsec:Jeps}}).

One then wonders how this could possibly circumvent the \emph{moment closure problem}, where a finite number of \emph{moments} ($p$-integrals of $f$) cannot reproduce kinetic theory exactly. 
An elementary example is the raw moment hierarchy, where one considers the raw moments $\mu_l = \int d^np \ p\ldots p f $ ($l$ times), and find that the time evolution of $\mu_l$ depends on $\mu_{l+1}$. 
Sadly, our generalised entropies \emph{do not} circumvent the moment closure problem. 
The Hamiltonian functional $H_{KT}[f] = \int d^nq d^np \ \lvert p \rvert^2 f/2$ for $1$-particle kinetic theory does not factor through any of the Poisson maps $\mathcal{J}_A$, and in particular, not through $\mathcal{J}_1$. 
%This can be attributed to the difference in the definition of temperature in kinetic theory and fluid dynamics. In kinetic theory, the temperature $\theta[f](q)$ is defined by the variance in the distribution of momentum at position $q$. In fluid dynamics however, taking $\rho, m$ and $s$ as our hydrodynamic variables, the temperature $T(\rho,s)$ is obtained by inverting the thermodynamic equation of state $s = s(\rho,T)$ for a monatomic ideal gas. 
%In section \textbf{\ref{sec:Hamiltonian}} we show that the two definitions of temperature agree if and only if the distribution function $f$ is a \emph{local Maxwellian} $f_m$, a distribution function defined by the moments $\int d^np \ (1,p,|p|^2) f$ (see (\ref{localMdef}), section \textbf{\ref{sec:Hamiltonian}}). 
%Since generic local Maxwellians do not stay as local Maxwellians under the collisionless Boltzmann equation or the Vlasov--Poisson equation, the Poisson map $\mathcal{J}_1$ never provides an exact reduction from kinetic theory to fluid dynamics. 
%Other approaches to moments and the moment closure problem are discussed in section \textbf{\ref{subsec:Hdiscuss}} \cite{Gibbons81, GibbonsHolmTronci08, Tassi17, PerinChandreMorrisonTassi15, Tassi15, Tassi16}.
While the Poisson map $\mathcal{J}_1$ only defers the moment closure problem to the Hamiltonian functional instead of solving it, we can still make progress by taking the perspective of \emph{approximate reduction}. We seek a physically motivated but otherwise uncontrolled approximation that allows the derivation of reduced equations. 
%In our case, the difference between $\theta[f]$ and $T(\rho[f],s[f])$ is controlled by the \emph{relative entropy} $r[f\vert f_m]$ of $f$ against its local Maxwellian $f_m$ \cite{Cercignani88,SaintRaymond09book}. This difference can be thought of the \enquote{non-hydrodynamic energy} in the kinetic system, and we show that this term stays small if the initial distribution function is close to a global Maxwellian in section \textbf{\ref{subsec:globalM}}. 
In section \textbf{\ref{sec:Hamiltonian}}, we obtain a decomposition $H_{KT} = H_{fluids} \circ \mathcal{J}_1 + \Delta H$, where $H_{fluids}$ is the fluid Hamiltonian that generates the compressible Euler equations for a monatomic ideal gas, and $\Delta H$ is a term that depends on the relative entropy density $r[f|f_m]$ of $f$ against its local Maxwellian $f_m$.
By ignoring $\Delta H$ in the Hamiltonian functional, we obtain a \emph{manifestly Hamiltonian} derivation of the compressible Euler equations from $1$-particle kinetic theory -- the Hamiltonian structure of the hydrodynamic equations is inherited from $1$-particle kinetic theory through the Poisson map $\mathcal{J}_1$. This type of derivation of reduced Hamiltonian models through uncontrolled approximations in the Hamiltonian functional, or equivalently in the Lagrangian density of a variational principle, is not uncommon in fluid dynamics and mathematical physics, and we give some examples in section \textbf{\ref{subsec:Hdiscuss}} \cite{MilesSalmon85, Salmon85, Whitham67,Whitham70,Whitham74book, MalomedWeinstein96, Kevrekidis09book, KhesinMisiolekModin20,KhesinMisiolekModin19, MarsdenMorrisonWeinstein84}.

We close this paper with an investigation on the term $\Delta H$. In section \textbf{\ref{subsec:globalM}}, we obtain a bound on the size of $\Delta H$ in the near global Maxwellian regime in terms of constants of motion for the collisionless Boltzmann equation (\ref{Boltzmann1}) and the Vlasov--Poisson equation (\ref{Vlasov Poisson}). In section \textbf{\ref{subsec:isotropic}}, we obtain an approximation of $\Delta H$ in terms of the generalised entropy variables for isotropic, near local Maxwellian distribution functions using a formal perturbation expansion.

\section{Symplectic and Poisson geometry}\label{sec:maths}

In this section we set out our notation and summarise a few well-known results in symplectic and Poisson geometry that we will make extensive use of later. The main references are \cite{Weinstein83, Cannas01book}. Readers familiar with symplectic and Poisson geometry can skip this section and only refer to it for the notation used in the subsequent sections. We will assume that all manifolds and other objects are smooth. 

\subsection{Basic definitions}\label{subsec:defs}
A \emph{symplectic manifold} $S$ is a $2n$-dimensional manifold together with a \emph{symplectic $2$-form} $\omega$ on $S$ that is closed and nondegenerate: $d\omega = 0$ and $\omega^n \rvert_x \neq 0$ for all $x \in S$. Given a function $f$ on $S$, its \emph{Hamiltonian vector field} $X_f$ is the unique solution to the equation
$\iota_{X_f}\omega = df$, where $\iota$ denotes the interior product. 
The \emph{canonical Poisson bracket} $\{\cdot,\cdot\} : C^\infty(S)\times C^\infty(S) \rightarrow C^\infty(S)$ on $S$ is defined as
\begin{align}
\{f,g\} = - \mathcal{L}_{X_f} g = \mathcal{L}_{X_g} f = \omega(X_f,X_g),
\end{align}
where $f$ and $g$ are functions on $S$, and $\mathcal{L}_X$ is the Lie derivative along the vector field $X$. The canonical Poisson bracket satisfies a list of important identities: for all $\alpha,\beta\in\mathbb{R}$ and all $f,g,h \in C^\infty(S)$,
\begin{align}
\label{Poisson_start}
\{\alpha f+\beta g,h\} &= \alpha \{f,h\}+\beta\{g,h\}, &\text{($\mathbb{R}$--\emph{bilinearity})} \\
\{f,g\} &= -\{g,f\}, &\text{(\emph{antisymmetry})} \\
\label{Leibniz}
\{f,gh\} &= \{f,g\}h + \{f,h\}g, &\text{(\emph{Leibniz identity})} \\
0 &= \{\{f,g\},h\} + \{\{g,h\},f\} + \{\{h,f\},g\}. &\text{(\emph{Jacobi identity})}
\label{Poisson_end}
\end{align}
Conversely, a manifold $P$ equipped with a binary operation $\{\cdot,\cdot\}_{NC}: C^\infty(P)\times C^\infty(P)\rightarrow C^\infty(P)$ satisfying (\ref{Poisson_start}-\ref{Poisson_end}) is called a \emph{Poisson manifold}, and the operation $\{\cdot,\cdot\}_{NC}$ is called a \emph{noncanonical Poisson bracket}. All symplectic manifolds are Poisson manifolds, but not necessarily vice versa. To distinguish between Poisson brackets on symplectic manifolds and non-symplectic Poisson manifolds, we will always put subscripts on the latter and leave the former unsubscripted.
The local theory of finite-dimensional Poisson manifolds is extensively studied in \cite{Weinstein83}. 
%On a finite-dimensional Poisson manifold $P$, the Poisson bracket can be written in terms of an antisymmetric contravariant rank-$2$ tensor $\Pi$, called the \emph{Poisson tensor}, such that $\{f,g\}_{NC} = \Pi(df,dg)$ for all functions $f,g$ on $P$. Now let $y:\mathbb{R}\rightarrow\mathbb{R}$ be a real function. If $f:P\rightarrow \mathbb{R}$ is a function on a finite-dimensional Poisson manifold $P$, then $y(f) = y \circ f$ is also a function on $P$. We have
An important identity on finite-dimensional symplectic and Poisson manifolds is that
\begin{align}\label{ydiff}
\{y(f),g\}_{NC} = y^\prime (f) \{f,g\}_{NC}.
\end{align}
%While the Leibniz identity directly implies (\ref{ydiff}) for polynomial $y$, the general case for smooth functions relies on the identification of Poisson brackets with Poisson tensors. 
We will make extensive use of (\ref{ydiff}) in sections \textbf{\ref{subsec:JA}} and \textbf{\ref{subsec:Jeps}}.

\subsection{Poisson brackets and affine functions on a cotangent bundle}\label{subsec:cotangent}

The \emph{cotangent bundle} $T^*M$ of an $n$-dimensional manifold $M$ is a $2n$-dimensional symplectic manifold that we will make extensive use of later. The points in $T^*M$ are pairs $(q,p)$, where $q$ is a point on $M$, and $p \in T^*_qM$ is a cotangent vector at $q$. 
%The cotangent vector spaces $T^*_qM$ are called the fibres of $T^*M$. 
On a coordinate neighbourhood $U$ of $q$, there are induced coordinates $(q^1,\ldots,q^n,p_1,\ldots,p_n)$ known as \emph{canonical coordinates} on $T^*U$, where $q^i$ are the coordinates of the basepoint $q$ and $p_i = \langle p, \partial/\partial q^i \rangle_q$ are the components of the cotangent vector $p$. Here $ \left\langle \cdot,\cdot\right\rangle_{q}: T_qM \times T_q^*M \rightarrow \mathbb{R}$ denotes the dual pairing. We will henceforth use the summation convention on the Latin indices $i,j = 1,\ldots n$.
In canonical coordinates, the local expression of the symplectic form on $T^*M$ is $\omega = dq^i \wedge dp_i$, and its Poisson bracket $\{\cdot,\cdot\}$ is given as
\begin{align}\label{classicalPB}
\{f,g\} = \frac{\partial f}{\partial q^i}\frac{\partial g}{\partial p_i} - \frac{\partial g}{\partial q^i}\frac{\partial f}{\partial p_i},
\end{align}
for functions $f,g$ on $T^*M$.

Some natural geometric constructions on $M$ can be formulated using the Poisson bracket (\ref{classicalPB}) on $T^*M$. Let $u = u^i{\partial}/{\partial q^i}$ be a vector field on $M$.
We can associate $u$ to a function $\left\langle u,p\right\rangle: (q,p) \mapsto \left\langle u(q),p\right\rangle_{q}$ on $T^*M$, which in canonical coordinates is $\left\langle u,p\right\rangle = u^i(q)p_i$.
Any function on $T^*M$ that is linear in $p$ arises from the dual pairing of $p$ with a vector field on $M$. We call these functions \emph{fibrewise linear}.

A function $g$ on $M$ defines a function $\pi^*g = g \circ \pi$ on $T^*M$ by pullback, where $\pi:T^*M\rightarrow M$ is the projection. In canonical coordinates, these functions are precisely those that depend on $q$ only: $g(q,p) = g(q)$. We call these function \emph{fibrewise constant}. By analogy with affine functions defined on a vector space, we call a function on $T^*M$ \emph{fibrewise affine} if it is a linear combination of fibrewise linear and fibrewise constant functions.
The space of fibrewise affine functions is a Lie subalgebra of the Lie algebra $(C^\infty(T^*M),\{\cdot,\cdot\})$ \cite{Guillemin80,Marsden83}. Let $u,v$ be vector fields on $M$ and $g,h$ be functions on $M$. We have
\begin{align}
\label{qpPB1}
\left\{ \langle u,p\rangle,\langle v,p\rangle \right\} &= - \left\langle [u,v],p\right\rangle , \\
\label{qpPB2}
\left\{ \langle u,p\rangle, \pi^*g \right\} &= - \pi^*\left( \mathcal{L}_u g \right), \\
\label{qpPB3}
\left\{ \pi^* g, \pi^* h \right\} &= 0,
\end{align}
where $[u,v] = \mathcal{L}_u v$ is the Lie bracket for vector fields on $M$, given in coordinates as
\begin{align}
\label{LieVect}
[u,v]^i = u^j\frac{\partial v^i}{\partial q^j} - v^j\frac{\partial u^i}{\partial q^j}.
\end{align}
While symmetric contravariant (upper index) tensors $T^{i_1\ldots i_k}(q)$ on $M$ induce fiberwise polynomial functions $T^{i_1\ldots i_k}(q)p_{i_1}\ldots p_{i_k}$ on $T^*M$, the quadratic and higher degree polynomials generate a subalgebra that necessarily includes polynomials of arbitrarily large degree. 
%The fibrewise affine functions constitute a maximal \enquote{small} Lie subalgebra of the fibrewise polynomials. 
%This is somewhat analogous to the situation on a symplectic vector space $\mathbb{R}^{2n}$, where the polynomials of degree at most $2$ constitute a maximal finite-dimensional Lie subalgebra of the Lie algebra $\mathrm{Pol}(\mathbb{R}^{2n})$ of polynomials under the canonical Poisson bracket.

The Lie algebra of fibrewise affine functions on the cotangent bundle can also be constructed as a semidirect product Lie algebra. The space of vector fields on $M$, denoted $\mathrm{Vect}(M)$, is a Lie algebra under the Lie bracket for vector fields (\ref{LieVect}). The action of vector fields on functions by Lie differentiation defines a representation of $\mathrm{Vect}(M)$ on $C^\infty(M)$. Thus we can construct the \emph{semidirect product Lie algebra} $\mathfrak{s}_0 = \mathrm{Vect}(M) \ltimes C^\infty(M)$. As a vector space, $\mathfrak{s}_0$ consists of pairs $(u,g)$ where $u$ is a vector field on $M$ and $g$ is a function on $M$. The Lie bracket on $\mathfrak{s}_0$ is
\begin{align}
\left[(u,g),(v,h)  \right]_{\mathfrak{s}_0} = \left( [u,v],\mathcal{L}_u h - \mathcal{L}_v g \right).
\end{align}
Now consider the linear map $\mathfrak{s}_0 \rightarrow C^\infty(T^*M)$ defined by
\begin{align}
(u,g) \mapsto \langle u,p \rangle + \pi^* g.
\end{align}
This map is injective, and its image is precisely the space of fibrewise affine functions on $T^*M$. The identities (\ref{qpPB1}-\ref{qpPB3}) shows that this map is an anti-homomorphism of Lie algebras, namely that
\begin{align}
\left\{ \langle u,p\rangle + \pi^* g, \langle v,p\rangle + \pi^* h\right\} = -\left( \left\langle [u,v],p \right\rangle + \pi^*\left(\mathcal{L}_u h- \mathcal{L}_v g \right)  \right).
\end{align}
Thus $\mathfrak{s}_0$ can be considered as the Lie subalgebra of fibrewise affine functions in $(C^\infty(T^*M),\{\cdot,\cdot\})$, up to a sign. 

\subsection{Integration on symplectic manifolds}\label{subsec:integration}

Integration on a symplectic manifold interacts with the canonical Poisson bracket in a very similar way to how the trace interacts with the matrix commutator. The $2n$-dimensional symplectic manifold $S$ has a standard non-vanishing $2n$-form $dV = \omega^n/n!$ called the \emph{symplectic volume form}. (Despite the notation, $dV$ is not necessarily an exact form.) For the cotangent bundle $T^*M$ of an $n$-dimensional manifold $M$, the local expression for $dV$ in canonical coordinates is
\begin{align}
dV = dq^1\wedge \ldots \wedge dq^n \wedge dp_1 \wedge \ldots \wedge dp_n.  
\end{align}
The volume form induces a weakly non-degenerate symmetric bilinear form on the space of functions on $S$ by integration:
\begin{align}
(f,g) = \int_S dV \ fg,
\end{align}
where $f,g$ are functions on $S$ that satisfy sufficient decay conditions. The condition \enquote{weakly non-degenerate} means that if $(f,g)=0$ for all sufficiently decaying functions $g$, then $f=0$. Moreover, since $\{f,g\}\omega^n$ is an exact $2n$-form, we have 
\begin{align}
\int_S dV \ \{f,g\} = 0, \quad \int_S dV \ \{f,g\}h = \int_S dV \ f\{g,h\}
\end{align} 
for any $f,g,h\in C^\infty(S)$, subject to decay conditions on the integrands. These identities for $S = T^*M$ will be used extensively in section \textbf{\ref{sec:Poissonmap}}.

%One way to make the decay conditions precise is to take an exhaustion of $S$ by compact submanifolds --  take a sequence $S_k$ of compact submanifolds with boundary such that $\cup S_k = S$ and $S_k \subset \mathrm{Int}(S_{k+1})$, where $\mathrm{Int}$ denotes the interior of a manifold with boundary. The decay conditions on $f$ can then be phrased in terms of the decay rates of the numerical sequence $\mathrm{sup}_{x\in S_{k+1} \setminus S_{k}}(\lvert f(x)\rvert)$, and similarly for its derivatives. For example, if $S=\mathbb{R}^{2n}$ is the classical phase space for a particle in $\mathbb{R}^n$, we can take $S_k$ to be the ball of radius $k$; if $S=T^*M$ where $M$ is a compact Riemannian manifold, we can take $S_k$ to be $\bigcup_q B_k(q)$, where $B_k(q)$ is the ball of radius $k$ on $T^*_qM$ under the norm induced by the Riemannian metric. We will henceforth assume that all functions satisfy sufficient decay conditions.

\subsection{The Lie--Poisson bracket on the dual of a Lie algebra}\label{subsec:LPbracket}

The dual space to a Lie algebra is an important example of a non-symplectic Poisson manifold. Let $\mathfrak{g}$ be a Lie algebra with Lie bracket $[\cdot,\cdot]$, which we assume to be finite-dimensional for now, and let $\mathfrak{g}^{*}$ be its dual vector space. There is a natural Poisson bracket on $\mathfrak{g}^{*}$, called the \emph{Lie--Poisson bracket}, that turns $\mathfrak{g}^{*}$ into a Poisson manifold. For functions $f(\mu),g(\mu)$ on $\mathfrak{g}^{*}$, it is given by
\begin{align}\label{LiePoissonbracket}
\{f,g\}_{\mathfrak{g}}(\mu) = \pm \left\langle \mu, \left[\frac{\partial f}{\partial \mu},\frac{\partial g}{\partial \mu} \right] \right\rangle_{\mathfrak{g}},
\end{align}
where the angle bracket denotes the dual pairing between $\mathfrak{g}^{*}$ and $\mathfrak{g}$, and we interpret ${\partial f}/{\partial \mu}$ to be taking values in $\mathfrak{g}$ instead of $\mathfrak{g}^{**}$. The bracket $\{\cdot,\cdot\}_{\mathfrak{g}}$ satisfies (\ref{Poisson_start}-\ref{Poisson_end}), so it is a Poisson bracket.
Geometrically, if $\mathfrak{g}$ is the Lie algebra of a Lie group $G$, we can also obtain the Lie--Poisson structure on $\mathfrak{g}^*$ reducing the symplectic structure on $T^*G$ by the natural left-(or right-)$G$ action \cite{Marsden83, Marsden84b, Marsden13book}.
If $\mathfrak{g}$ is infinite-dimensional, we take $\mathfrak{g}^*$ to be a \emph{smooth dual space} of $\mathfrak{g}$, a subspace of the topological dual space that has a weakly non-degenerate pairing with $\mathfrak{g}$. The Lie--Poisson bracket (\ref{LiePoissonbracket}) still makes sense on $\mathfrak{g}^*$ provided we replace the functions on $\mathfrak{g}^*$ with sufficiently regular functionals, and the ordinary derivatives with functional derivatives. Rigorous analytical treatments of Lie--Poisson dynamics on infinite-dimensional spaces can be found in \cite{EbinMarsden70, ChernoffMarsden74book, Khesin08book}.

%In applications to fluid dynamics and kinetic theory, the relevant Lie algebras are usually infinite-dimensional. Functions on these infinite-dimensional Lie algebras become \emph{functionals}, and the ordinary derivative on these spaces has to be replaced with the \emph{functional derivative}. We proceed formally by taking \enquote{smooth dual spaces} \cite{Khesin08book} and restricting our attention to sufficiently regular functionals. A smooth dual space is a subspace of the topological dual space that has a weakly non-degenerate pairing with the original space. The choice of a smooth dual is non-unique and is usually motivated by context. For example, the space of functions on a symplectic manifold $C^\infty(S)$ is an infinite-dimensional Lie algebra under the canonical Poisson bracket $\{\cdot,\cdot\}$. The smooth dual of $C^\infty(S)$ can be taken to be the space of smooth volume forms $\Omega^{2n}(S)$, perhaps with sufficient decay conditions, with the dual pairing given by integration. If we restrict our attention to sufficiently regular functionals on $\Omega^{2n}(S)$ so that their functional derivatives exist and can be considered to take values in $C^\infty(S)$, we can formally consider $\Omega^{2n}(S)$ to be an \enquote{infinite-dimensional Poisson manifold} equipped with the Lie--Poisson bracket from $C^\infty(S)$. 

\section{Noncanonical Hamiltonian mechanics, Poisson maps and reduction}\label{sec:physics}

In this section we briefly review \emph{noncanonical Hamiltonian mechanics} and outline the \emph{reduction of noncanonical Hamiltonian systems} that we will pursue in the subsequent sections. Detailed expositions on noncanonical Hamiltonian mechanics with an emphasis on fluid dynamics and kinetic theory can be found in \cite{Morrison98, Salmon88, Tassi17}.

Consider the configuration space $\mathfrak{M}$ for a set of dynamical variables. In physical applications, $\mathfrak{M}$ is often an infinite-dimensional space of \emph{field variables}, i.e. a space of sections of a given vector bundle over a finite-dimensional manifold. Let $\mathcal{F}(\mathfrak{M})$ denote a space of sufficiently regular functionals $F:\mathfrak{M}\rightarrow \mathbb{R}$. We think of functionals $F \in \mathcal{F}(\mathfrak{M})$ as observable quantities on the filed variables in $\mathfrak{M}$.
Now let $\mathfrak{M}$ be a Poisson manifold, with Poisson bracket $\{\cdot,\cdot\}_{\mathfrak{M}}$. Define the \emph{noncanonical Hamiltonian system} on $\mathfrak{M}$ generated by a \emph{Hamiltonian functional} $H\in \mathcal{F}(\mathfrak{M})$ to be the dynamical system where any $F \in \mathcal{F}(\mathfrak{M})$ evolves in time as
\begin{align}
\frac{d{F}}{dt} = \{F,H\}_{\mathfrak{M}}.
\end{align}
If $\mathfrak{M}$ is a space of fields, letting $F[z] = \langle \psi , z \rangle_\mathfrak{M}$ for arbitrary test functions $\psi$ recovers the equation of motion for the state $z \in \mathfrak{M}$ in a weak sense.
For the rest of the paper, each configuration space considered is the smooth dual space to an infinite-dimensional Lie algebra equipped with the Lie--Poisson bracket (see section \textbf{\ref{subsec:LPbracket}}). Such configuration spaces are ubiquitous in fluid dynamics \cite{Marsden84a, Marsden84b, Morrison80, Holm85} and kinetic theory \cite{Marsden83, Holm85}.
%For example, the kinetic theory of dilute gases or electrostatic plasmas can be formulated on the dual of $C^\infty(T^*M)$, where the Lie algebra structure is given by the canonical Poisson bracket \cite{Marsden83, Holm85}. Ideal compressible hydrodynamics and magnetohydrodynamics can be formulated on the duals of appropriate semidirect product Lie algebras of $\mathrm{Vect}(M)$ \cite{Marsden84a, Marsden84b, Morrison80, Holm85}. The Hamiltonian functional is the total energy of the system in each case.

\subsection{Poisson maps and reduction}\label{subsec:myreduction}

Given a noncanonical Hamiltonian system $(\mathfrak{M}_1,\{\cdot,\cdot\}_1,\tilde{H})$, we can determine the time evolution of all functionals $\tilde{F} \in \mathcal{F}(\mathfrak{M}_1)$. However, $\mathfrak{M}_1$ may have unimportant degrees of freedom. We would like to reformulate the noncanonical Hamiltonian system on a reduced configuration space $\mathfrak{M}_2$ that contains only the relevant degrees of freedom. We formalise this notion using Poisson manifolds and Poisson maps.

Given Poisson manifolds $(\mathfrak{M}_1,\{\cdot,\cdot\}_1)$ and $(\mathfrak{M}_2,\{\cdot,\cdot\}_2)$, we say that a map $\mathcal{J}:\mathfrak{M}_1\rightarrow\mathfrak{M}_2$ is a \emph{Poisson map}, if for all functionals $F,G \in \mathcal{F}(\mathfrak{M}_2)$,
\begin{align}\label{JPoisson}
\left\{ F\circ \mathcal{J}, G\circ \mathcal{J} \right\}_1 = \left\{ F, G \right\}_2 \circ \mathcal{J},
\end{align}
where $F\circ \mathcal{J}$ and $G\circ \mathcal{J}$ are now functionals on $\mathfrak{M}_1$ by composition. The pullback map $F \mapsto F \circ \mathcal{J}$ is commonly denoted as $\mathcal{J}^*:\mathcal{F}(\mathfrak{M}_2)\rightarrow \mathcal{F}(\mathfrak{M}_1)$. 
In the noncanonical Hamiltonian system $(\mathfrak{M}_1,\{\cdot,\cdot\}_1,\tilde{H})$, suppose we are only interested in the time evolution of functionals $\tilde{F} = F \circ \mathcal{J}$ for some $F\in\mathcal{F}(\mathfrak{M}_2)$.
If the Hamiltonian functional also factors through $\mathcal{J}$ i.e. $\tilde{H} \in \mathcal{F}(\mathfrak{M}_1)$ is of the form
\begin{align}\label{Hcollective}
\tilde{H} = H \circ \mathcal{J} \qquad \text{for some $H\in\mathcal{F}(\mathfrak{M}_2)$,}
\end{align}
then, since $\mathcal{J}$ is a Poisson map,
\begin{align}
\frac{d{\tilde{F}}}{dt}  = \left\{ \tilde{F}, \tilde{H} \right\}_1 = \left\{ F, H \right\}_2 \circ \mathcal{J},
\end{align}
so we can replace the evolution equation for $\tilde{F}$ with an evolution equation $\dot{F} = \{F,H\}_{2}$ for $F$, and forget about $\mathfrak{M}_1$ altogether. A Hamiltonian functional that factors through a Poisson map is called \emph{collective} in \cite{Guillemin80}.

The adjoint linear map of a Lie algebra homomorphism is an important example of a Poisson map. Let $\mathfrak{g}$ and $\mathfrak{h}$ be Lie algebras, and $\phi:\mathfrak{g}\rightarrow\mathfrak{h}$ be a Lie algebra (anti)-homomorphism. The adjoint linear map $\phi^* : \mathfrak{h}^* \rightarrow \mathfrak{g}^*$ is a Poisson map if we equip the duals spaces $\mathfrak{g}^*$ and $\mathfrak{h}^*$ with the Lie--Poisson bracket of the same (opposite) sign. In fact, the adjoint of a linear map $\phi:\mathfrak{g}\rightarrow\mathfrak{h}$ is a Poisson map if and only if $\phi$ is a Lie algebra homomorphism (\cite{Marsden83}, Lemma 8.2). (This result still holds formally if the Lie algebras are infinite-dimensional.) The Poisson map from $1$-particle kinetic theory to fluids with constant entropy is precisely the adjoint linear map to the inclusion of $\mathfrak{s}_0 = \mathrm{Vect}(M) \ltimes C^\infty(M)$ into $C^\infty(T^*M)$ as the Lie subalgebra of fibrewise affine functions met in section \textbf{\ref{subsec:cotangent}} \cite{Guillemin80, Marsden83}.

To carry out our programme of reduction, we construct a Poisson map from the configuration space for $1$-particle kinetic theory to that of ideal compressible fluids with non-constant entropy (section \textbf{\ref{sec:Poissonmap}}). However, the kinetic-theory Hamiltonian is not the pullback of an effective Hamiltonian in the hydrodynamic variables, so we make an \enquote{approximation} to the kinetic-theory Hamiltonian to make it so (section \textbf{\ref{sec:Hamiltonian}}). The rest of this paper is devoted to carry out this \emph{approximate reduction}.

\subsection{Other forms of Hamiltonain reduction}\label{subsec:notmyreduction}

There are other methods that can be used to construct Hamiltonian systems with fewer degrees of freedom than a more primitive system. Here we give some examples and compare them with those developed in section \textbf{\ref{subsec:myreduction}}. We will not pursue these methods any further in sections \textbf{\ref{sec:Poissonmap}} and \textbf{\ref{sec:Hamiltonian}}.

The classic example is \emph{Marsden--Weinstein reduction} of symplectic or Poisson manifolds \cite{Meyer73, MontgomeryMarsdenRatiu84, MarsdenWeinstein74, MarsdenRatiu86, Cannas01book, Marsden83}. Let $\mathfrak{M}$ be a symplectic or Poisson manifold, $G$ be a Lie group and $\mathfrak{g}$ be the Lie algebra of $G$. Suppose that we have a Hamiltonian $G$-action on $M$, with \emph{(equivariant) momentum map} $\mathcal{J}: \mathfrak{M} \rightarrow \mathfrak{g}^*$. In particular, $\mathcal{J}$ is a Poisson map such that the function $z \rightarrow \langle \mathcal{J}(z),X\rangle_\mathfrak{g}$ generates the flow of the one-parameter group $\exp(tX) \subset G$ on $\mathfrak{M}$ for all $X \in \mathfrak{g}$ \cite{Marsden13book, Guillemin80}. Under suitable conditions, the \emph{Marsden--Weinstein quotient} $\mathcal{J}^{-1}(0)/G$ is correspondingly a symplectic or Poisson manifold. A $G$-invariant Hamiltonian functional on $\mathfrak{M}$ induces a reduced Hamiltonian system on $\mathcal{J}^{-1}(0)/G$. Marsden--Weinstein reduction achieves a similar goal to the reduction illustrated in section \textbf{\ref{subsec:myreduction}}, where unimportant degrees of freedom are forgotten to form a reduced Hamiltonian system, but it uses the Poisson map $\mathcal{J}$ in a different way.

More recently, various Hamiltonian \enquote{reduced fluid models} for drift kinetics are derived by constructing a new but related Poisson bracket on the space of reduced variables in \cite{Tassi15, Tassi16, Tassi17}. A similar construction is used to find reduced descriptions for the one-dimensional Vlasov--Amp\'{e}re system in \cite{PerinChandreMorrisonTassi15}. 
We give a brief summary of the techniques involved. Let $(\mathfrak{M}_1,\{\cdot,\cdot\}_1)$ be the Poisson manifold of the more primitive variables, $\mathfrak{M}_2$ be the space of reduced variables, and $\pi: \mathfrak{M}_1 \rightarrow \mathfrak{M}_2$ be a map from primitive to reduced variables.
In these applications, given two functionals $F,G$ on $\mathfrak{M}_2$, the Poisson bracket of their pullbacks $\left\{ \pi^*F, \pi^*G \right\}_1$ cannot be expressed as the pullback of some other functional on $\mathfrak{M}_2$. There are excess variables not described by the image $\pi(\mathfrak{M}_1)$ that appear in the expression of $\left\{ \pi^*F, \pi^*G \right\}_1$.
These excess variables are eliminated by imposing a \emph{closure relation}, which is an additional functional relationship between the excess and reduced variables.
We can describe this as an embedding $\sigma: \mathfrak{M}_2 \rightarrow \mathfrak{M}_1$, whose image corresponds to the locus of the closure relation on $\mathfrak{M}_1$. We require that $\pi \circ \sigma = \mathrm{id}_2$, the identity map on $\mathfrak{M}_2$. 
Define a bracket $\{\cdot,\cdot\}_2$ on $\mathfrak{M}_2$ by
\begin{align}\label{pinonPoisson}
\left\{ F, G \right\}_2 = \sigma^* \left\{ \pi^* F, \pi^*G \right\}_1,
\end{align}
where $F$ and $G$ are arbitrary functionals on $\mathfrak{M}_2$. The bracket $\{\cdot,\cdot\}_2$ always satisfies (\ref{Poisson_start} - \ref{Leibniz}). If the closure relation $\sigma$ is so chosen that $\{\cdot,\cdot\}_2$ satisfies the Jacobi identity (\ref{Poisson_end}), ($\mathfrak{M}_2$,$\{\cdot,\cdot\}_2$) becomes a Poisson manifold. (The Jacobi identity for $\{\cdot,\cdot\}_2$ is checked manually in \cite{Tassi15,Tassi16,PerinChandreMorrisonTassi15}.) 
While $\sigma^{*}\pi^{*} = \mathrm{id}_2^{*}$, the composition $\pi^{*}\sigma^{*}$ is not the identity for functionals on $\mathfrak{M}_1$.
The relation (\ref{pinonPoisson}) is not the same as (\ref{JPoisson}), and in particular $\pi$ is \emph{not} a Poisson map.
 
In the applications described in \cite{Tassi15,Tassi16,PerinChandreMorrisonTassi15}, the primitive Hamiltonian $\tilde{H}$ generating the dynamics on $\mathfrak{M}_1$ is the pullback of some effective Hamiltonian $H$ on $\mathfrak{M}_2$ through $\pi$, that is, $ \tilde{H} = \pi^{*}H$.
The reduced dynamics is then given by
\begin{align}\label{notmyreductioneq}
\frac{d F}{dt} = \left\{ F, H \right\}_2 = \sigma^* \left\{ \pi^*F, \tilde{H} \right\}_1,
\end{align}
so the functional $F$ evolves according to the value of $\left\{ \pi^*F, \tilde{H} \right\}_1$ restricted on $\sigma(\mathfrak{M}_2)$. This framework provides a manifestly Hamiltonian way to impose certain \enquote{good} closure relations to obtain reduced evolution equations, which is often an improvement over imposing the closure relations directly in the primitive evolution equations.

In a separate class of applications, one seeks physically relevant Poisson submanifolds of the Poisson manifold of primitive variables. The Hamiltonian functional restricted to the submanifold then generates a reduced Hamiltonian system. Such reductions aim to obtain simpler descriptions of the primitive system under special conditions.
For example, the extended magnetohydrodynamics (MHD) system, which is an extension to ordinary MHD that incorporates the Hall effect \cite{Schnack09book}, is known to be a noncanonical Hamiltonian system \cite{AbdelhamidKawazuraYoshida15}. In \cite{KaltsasThroumoulopoulosMorrison17}, a Poisson submanifold corresponding to extended MHD configurations that are translationally invariant along the $z$-axis is found, and the corresponding reduced Hamiltonian system is studied in detail.

Another example in kinetic theory is the waterbag reduction of the Benney hydrodynamic chain \cite{ChesnokovPavlov12,Yu2000}. The Benney hydrodynamic chain \cite{Benney73} is a system of evolution equations that can be obtained as the moment equations of a one-dimensional Vlasov-like system. Let $a_1(q,t), \ldots,a_K(q,t)$ be hydrodynamic variables, with $a_1<a_2<\ldots a_K$, and consider the $K$-waterbag distribution functions of the form
\begin{align}\label{waterbagf}
f(q,p,t) = \sum_{k=1}^{K-1} f_k \mathbf{1}_{[a_k(q,t),a_{k+1}(q,t)]}(p),
\end{align}
where the $f_k$ are constants and $\mathbf{1}_{[a,b]}$ is the indicator function on the interval $[a,b]$. The $K$-waterbag distribution functions (\ref{waterbagf}) constitute a Poisson submanifold, and there is a reduced Hamiltonian system on the hydrodynamic variables $a_k(q,t)$ \cite{ChesnokovPavlov12,Yu2000}. This technique has been extended to obtain reduced descriptions in other kinetic systems \cite{ChesnokovPavlov12,Pavlov14}.

\section{The Poisson map from kinetic theory to fluid dynamics}\label{sec:Poissonmap}

In this section we describe the formulation of kinetic theory and fluid dynamics as noncanonical Hamiltonian systems \cite{Marsden84a, Morrison80, Marsden83, KhesinMisiolekModin20,Holm85}, and obtain a Poisson map between the underlying Poisson manifolds. 

In the following, $M$ is an oriented $n$-dimensional manifold, and $T^*M$ is its cotangent bundle. The assumption on orientability is inessential and can be removed by replacing all top-degree differential forms with smooth densities. Various smoothness and decay conditions are implicitly assumed. For physical applications, it suffices to take $M = \mathbb{T}^n$ (or $\mathbb{R}^n$) and $T^*M = \mathbb{T}^n \times \mathbb{R}^n$ (respectively $\mathbb{R}^{2n}$), where $M$ is endowed with the standard metric $\delta_{ij}$ and standard volume element $d^n q$.

\subsection{The Lie--Poisson bracket for kinetic theory}\label{subsec:LPforKT}

Consider the space of functions $C^\infty(T^*M)$ on the cotangent bundle $T^*M$, which is a Lie algebra under the canonical Poisson bracket $\{\cdot,\cdot\}$ on $T^*M$. Its dual space is the space of volume forms on $T^*M$, which we denote as $\Omega^{2n}(T^*M)$. We use the symplectic volume element $dV = \omega^n /n!$ to identify $C^\infty(T^*M)$ with $\Omega^{2n}(T^*M)$ by the one-to-one correspondence $f \mapsto f dV$ . The dual pairing
\begin{align}
\left\langle f dV, g\right\rangle_{C^\infty(T^*M)} = \int_{T^*M} dV \ fg = (f,g)
\end{align}
gives a weakly non-degenerate symmetric invariant bilinear form on the Lie algebra $C^\infty(T^*M)$ (see section \textbf{\ref{subsec:integration}}). We interpret the functions $f\in C^\infty(T^*M)$ as the $1$-particle distribution functions on phase space in kinetic theory. The ($+$)-Lie--Poisson bracket is
\begin{align}
\{ F,G\}_{KT} [f] = \int_{T^*M} dV \ f \left\{ \frac{\delta F}{\delta f}, \frac{\delta G}{\delta f} \right\},
\end{align}
where $F[f],G[f]$ are functionals of $f$. Given a Hamiltonian functional $H[f]$, the evolution equation for $f$ is the familiar Liouville equation:
\begin{align}\label{KT-EOM}
\frac{\partial f}{\partial t} = \left\{ \frac{\delta H}{\delta f}, f \right\}.
\end{align}
For $M = \mathbb{T}^n$ or $\mathbb{R}^n$, choosing the Hamiltonian functional
\begin{align}
H_{KT}[f] = \int d^nqd^np \ \frac{\lvert p \rvert^2}{2} f
\end{align}
recovers the collisionless Boltzmann equation for a dilute gas. The Vlasov--Poisson equation for an electrostatic plasma can be obtained from the Vlasov--Poisson Hamiltonian \cite{Marsden83, ZakharovKuznetsov97, Holm85}:
\begin{align}\label{HVlasov}
H_{VP}[f] = \int d^nqd^np \ \frac{\lvert p \rvert^2}{2} f(q,p) + \frac{e^2}{2} \int d^n q d^np d^n q^\prime d^n p^\prime f(q,p)f(q^\prime,p^\prime)G(q,q^\prime) ,
\end{align}
where $e$ is the nondimensionalised charge of the particles. We write the potential energy using the Green's function $G(q,q^\prime)$ for the Laplacian $-\nabla^2$ on $\mathbb{T}^n$ or $\mathbb{R}^n$. This avoids having to introduce an explicit equation for the electrostatic potential
\begin{align}
\varphi(q) = e \int d^n q^\prime d^n p^\prime \ f(q^\prime,p^\prime)G(q,q^\prime).
\end{align}
For a plasma on $\mathbb{T}^n$, we assume that there is a uniformly charged inert background to keep the plasma neutral overall. 

The Jeans equation for a self-gravitating stellar system in galactic dynamics can be obtain from a slightly different Hamiltonian \cite{BinneyTremaine11book,Jeans15}
\begin{align}\label{HJeans}
H_{SG}[f] = \int d^nqd^np \ \frac{\lvert p \rvert^2}{2} f(q,p) - \frac{\mathrm{G}}{2} \int d^n q d^np d^n q^\prime d^n p^\prime f(q,p)f(q^\prime,p^\prime)G(q,q^\prime) ,
\end{align}  
where $\mathrm{G}$ is the dimensionless gravitational constant. The Jeans equation is identical to the Vlasov--Poisson equation, except that $e^2$ has been replaced with $-\mathrm{G}$. The sign difference leads to qualitative differences between the respective solutions of the Vlasov--Poisson and Jeans equations \cite{Glassey96bookchapter4}. The self-gravitating Hamiltonian $H_{SG}[f]$ is not bounded from below unlike $H_{KT}[f]$ and $H_{VP}[f]$. This difference becomes important in section \textbf{\ref{subsec:globalM}}.

Since the distribution function $f(q,p)$ represents the density of particles in phase space, the appropriate configuration space for distribution functions in kinetic theory should be taken as the space of \emph{positive} distribution functions $C^\infty(T^*M)_+$, where
\begin{align}
C^\infty(T^*M)_+ = \left\{  f \in C^\infty(T^*M) : f(q,p)>0 \ \text{for all $(q,p)\in T^*M$} \right\}.
\end{align} 
The strict positivity of $f$ also means that we can make sense of expressions such as $f\log f$ and its $f$-derivatives without having to define them as limits as $f \rightarrow 0$. $C^\infty(T^*M)_+$ is a Poisson submanifold because the Liouville equation (\ref{KT-EOM}) preserves the condition $f(q,p)>0$ for all $(q,p)$ for any choice of Hamiltonian functional.

\subsection{The Lie--Poisson bracket for ideal fluid dynamics with non-constant entropy}\label{subsec:LPforfluids}

The underlying Poisson manifold of the noncanonical Hamiltonian formulation of ideal fluid dynamics is the dual space of a semidirect product Lie algebra. Let $A$ be a non-negative integer. Consider the semidirect product Lie algebra $\mathfrak{s}_A = \mathrm{Vect}(M) \ltimes C^\infty(M,\mathbb{R}^{A+1})$, where the vector fields $u\in \mathrm{Vect}(M)$ act on the $\mathbb{R}^{A+1}$-valued functions $(g_0(q),\ldots,g_A(q)) \in C^\infty(M,\mathbb{R}^{A+1})$ by componentwise Lie differentiation. We will not use the summation convention on $\mathbb{R}^{A+1}$ and instead display the indices $a=0,\ldots,A$ explicitly. We will work implicitly with a standard Euclidean basis on $\mathbb{R}^{A+1}$ throughout. The semidirect product Lie bracket on $\mathfrak{s}_A$ is
\begin{align}
\left[ \left(u, (g_0,\ldots g_A) \right), \left(v, (h_0,\ldots h_A)\right) \right]_{\mathfrak{s}_A} = \left( [u,v], \left( \mathcal{L}_u h_0 - \mathcal{L}_v g_0, \ldots, \mathcal{L}_u h_A - \mathcal{L}_v g_A  \right) \right).
\end{align}
The dual space $\mathfrak{s}^*_A$ of the Lie algebra $\mathfrak{s}_A$ is $\mathrm{Vect}^*(M) \times \Omega^n(M,\mathbb{R}^{A+1})$, where $\mathrm{Vect}^*(M) = \Gamma(T^*M \otimes \wedge^n T^*M)$ is the space of $1$-form densities on $M$, and $\Omega^n(M,\mathbb{R}^{A+1})$ is the space of $\mathbb{R}^{A+1}$-valued volume forms on $M$, which we can also think of as the space of $A+1$ real-valued volume forms by taking components. A typical element of $\mathfrak{s}^*_A$ can be written as $\left(m_idq^i \otimes d^nq, (s_0d^nq, \ldots, s_Ad^nq)\right)$, where $m = m_i dq^i \otimes d^nq$ is a $1$-form density, and $s_ad^nq$ are top-degree forms for $a = 0,\ldots,A$. The dual pairing of $\mathfrak{s}^*_A$ with $\mathfrak{s}_A$ is given by integration:
\begin{align}
\left\langle \left(m_i dq^i \otimes d^nq, (s_0d^nq, \ldots, s_Ad^nq)\right),  \left(u, (g_0,\ldots, g_A) \right) \right\rangle_{\mathfrak{s}_A} = \int_M d^n q \ \left(m_iu^i + \sum_{a=0}^A s_a g_a \right).
\end{align}
The ($-$)-Lie--Poisson bracket on $\mathfrak{s}^*_A$ is given, for functionals $F,G\in\mathcal{F}(\mathfrak{s}^*_A)$, by
\begin{align}\label{sA_LP}
\{F,G\}_{\mathfrak{s}^*_A} = - \left\langle m,  \left[\frac{\delta F}{\delta m},\frac{\delta G}{\delta m} \right] \right\rangle_{\mathrm{Vect}(M)} - \sum_{a=0}^A \left\langle s_a, \mathcal{L}_{\frac{\delta F}{\delta m}} \frac{\delta G}{\delta s_a} - \mathcal{L}_{\frac{\delta G}{\delta m}} \frac{\delta F}{\delta s_a} \right\rangle_{C^\infty(M)} ,
\end{align}
where the angle brackets denote the pairing between the vector space indicated by the subscript and its dual. For $A=1$, we can interpret $m$ as the \emph{momentum density} of a fluid, $s_0 = \rho$ as the \emph{mass density} of a fluid, and $s_1 = s$ as the \emph{entropy density} of a fluid. For $M = \mathbb{T}^n$ or $\mathbb{R}^n$, the Lie--Poisson bracket (\ref{sA_LP}) has the explicit form
\begin{align*}
\{F,G\}_{\mathfrak{s}^*_1} = - \int d^n q & \quad m_i \left(\frac{\delta F}{\delta m_j}\frac{\partial}{\partial q^j}\frac{\delta G}{\delta m_i} - \frac{\delta G}{\delta m_j}\frac{\partial}{\partial q^j}\frac{\delta F}{\delta m_i} \right) \\
& +s_0 \left(\frac{\delta F}{\delta m_j}\frac{\partial}{\partial q^j}\frac{\delta G}{\delta \rho} - \frac{\delta G}{\delta m_j}\frac{\partial}{\partial q^j}\frac{\delta F}{\delta \rho} \right) \\
& + s_1\left(\frac{\delta F}{\delta m_j}\frac{\partial}{\partial q^j}\frac{\delta G}{\delta s} - \frac{\delta G}{\delta m_j}\frac{\partial}{\partial q^j}\frac{\delta F}{\delta s} \right),
\end{align*}
which coincides with the bracket for ideal fluids with non-constant entropy \cite{Morrison80}. The Hamiltonian functional that generates the compressible Euler equations is
\begin{align}\label{Hfluids0}
H_{fluids}[m,\rho,s] = \int d^nq \ \left(\frac{\lvert m\rvert^2}{2\rho} + \rho U(\rho,s)\right),
\end{align}
where $U$ is the internal energy density of the fluid as a local function of $\rho$ and $s$, given by some equation of state.

We have chosen to work with a general $A$ instead of fixing $A=1$. This becomes useful in section \textbf{\ref{subsec:JA}}, where we show that that there is a family of Poisson maps $\mathcal{J}_A$ from $C^\infty(T^*M)_+$ to $\mathfrak{s}^*_A$ for $A = 0,1,2,\ldots$, with the property that $\mathcal{J}_A$ can be obtained by truncating $\mathcal{J}_B$ for some $B \geq A$. The hydrodynamic mass and entropy densities can thus be thought of as the first two members of a series of \emph{generalised entropies}, as we shall explain in section \textbf{\ref{subsec:JA}}.

\subsection{The Poisson map from $C^\infty(T^*M)_+$ to $\mathfrak{s}^*_A$}\label{subsec:JA}

Now we construct the map $\mathcal{J}_A: C^\infty(T^*M)_+ \rightarrow \mathfrak{s}^*_A$, which we will soon prove to be a Poisson map. It is defined by the dual pairing
\begin{align}\label{Jdef}
\left\langle \mathcal{J}_A[f], \left(u, (g_0,\ldots,g_A) \right) \right\rangle_{\mathfrak{s}_A} = \int_{T^*M} dV \ \left( \langle u,p\rangle f + \sum_{a=0}^A f(\log f)^a g_a \right).
\end{align}
In components, we have $\mathcal{J}_A[f] = (m_i[f]dq^i\otimes d^nq , (s_0[f] d^n q, \ldots, s_A[f] d^n q))$, where
\begin{align}\label{Jdefcpts}
m_i[f]  =  \int d^n p \ p_i f, \qquad s_a[f]  =  \int d^n p \ f(\log f)^a
\end{align}
are the \emph{momentum density} and the \emph{generalised entropy densities} of the distribution function, respectively, obtained by integrating along the cotangent fibres i.e. integrating over $p$. In particular, $s_0[f]$ is the mass density and $s_1[f]$ is the spatial density of the Boltzmann entropy $f\log f$. Here and henceforth we use the sign convention that the entropy density $s_1[f]$ is a convex function of $f$ -- the opposite sign convention is more common in physics.

Now we state the main result of this paper:
\begin{prop}\label{JPoissonprop}
The map $\mathcal{J}_A: C^\infty(T^*M)_+ \rightarrow \mathfrak{s}^*_A$ (\ref{Jdef}) is a Poisson map from the space of positive $1$-particle distribution functions $C^\infty(T^*M)_+$ (as defined in section \textbf{\ref{subsec:LPforKT}}) to the space $\mathfrak{s}^*_A$ of hydrodynamic variables with generalised entropies (as defined in section \textbf{\ref{subsec:LPforfluids}}).
\end{prop}
The case for $A=0$ is known in \cite{Guillemin80, Marsden83} and can be immediately deduced from the fact that Lie algebra $\mathfrak{s}_0 = \mathrm{Vect}(M) \ltimes C^\infty(M)$ acts on $T^*M$ by Hamiltonian vector fields generated by the associated fibrewise affine function. The novelty here is that the Poisson map can be extended to account for the non-constant entropy bracket \cite{Morrison80}, and that $\mathcal{J}_A$ maps the $p$-integral of the Boltzmann entropy onto the hydrodynamic entropy density. We also find that there is family of spaces with hydrodynamic variables, all of which arises naturally from kinetic theory, that includes fluids with constant entropy as its first member and fluids with non-constant entropy as its second member. Taking the $p$-integral of the Boltzmann entropy to obtain a hydrodynamic entropy density is not a new idea, even within the framework of noncanonical Hamiltonian systems \cite{Grmela14}. In Porpotision \ref{JPoissonprop}, we settle the status of $s_1[f]$ as a component of a Poisson map. This offers a partial passage from kinetic theory to fluid dynamics by offering a Poisson map between the respective configuration spaces. However, as we shall see in section \textbf{\ref{sec:Hamiltonian}}, the reduction of noncanonical Hamiltonian systems described in section \textbf{\ref{sec:physics}} cannot be carried out exactly, because the kinetic-theory Hamiltonian does not factor through the Poisson map.

\begin{proof}
Define the functions $y_a(x): \mathbb{R}_{+} \rightarrow \mathbb{R}$ for $a=0,1,2,\ldots$ as
\begin{align}\label{yadef}
y_a(x) = x(\log x)^a.
\end{align} 
We set $(\log x)^0 = 1$ by convention so that $y_0(x)=x$, and also set $y_{a} = 0$ for $a < 0$. These functions have removable singularities at $x=0$ for $a \geq 1$. Each generalised entropy density $s_a[f]$ can be written in terms of the $y_a$ as $s_a[f]= \int d^n p \ y_a(f)$.
The functions $y_a(x)$ satisfy the recurrence relations
\begin{align}\label{yarecur1}
x \frac{d y_a}{dx} - y_a &= a y_{a-1}, \\
x \frac{d^2 y_a}{dx^2} &= a \frac{d y_{a-1}}{dx}.\label{yarecur2}
\end{align}
Conversely, any other set of functions $\tilde{y}_a$ that solves (\ref{yarecur1}) with $\tilde{y}_0(x)=x$ is of the form
\begin{align}\label{yaCa-b}
\tilde{y}_a(x) = y_a(x) + \sum_{b=0}^{a-1} C_{a-b} \frac{a!}{b!} y_b(x),
\end{align}
where $C_1,\ldots,C_A$ are arbitrary constants.
%where $\Lambda_{ba}$ is a strictly upper-triangular $(A+1)\times(A+1)$ matrix of constants, i.e. $\Lambda_{ba} = 0$ if $b \geq a$. We will discuss the transformation $y_a \mapsto \tilde{y}_a$ (\ref{yaLambda}) for arbitrary $\Lambda_{ba}$ at the end of the section.

Now suppose we have functionals $\tilde{F}, \tilde{G} \in \mathcal{F}(C^\infty(T^*M)_+)$ such that $\tilde{F}=F\circ \mathcal{J}_A$ and $\tilde{G}=G\circ \mathcal{J}_A$ for some $F,G \in \mathcal{F}(\mathfrak{s}^*_A)$. To prove that
\begin{align}
\left\{ \tilde{F}, \tilde{G} \right\}_{KT} = \left\{ F, G \right\}_{\mathfrak{s}_A^*} \circ \mathcal{J}_A,
\end{align}
we need to compute the functional derivative of $\tilde{F}=F\circ \mathcal{J}_A$. The functional chain rule gives
\begin{align}
\left\langle \frac{\delta \tilde{F}}{\delta f},\delta f\right\rangle_{C^\infty(T^*M)} = \left\langle \frac{\delta F}{\delta m}, \delta m[f]\right\rangle_{\mathrm{Vect}(M)} + \sum_{a=0}^A\left\langle \frac{\delta F}{\delta s_a}, \delta s_a[f]\right\rangle_{C^\infty (M)},
\end{align}
and
\begin{align}
\delta m[f] = d^n q\int d^n p \ p \delta f, \qquad \qquad
\delta s_a[f] = d^n q\int d^n p \ y^\prime_a(f)\delta f.
\end{align}
Comparing both expressions gives
\begin{align}
\frac{\delta \tilde{F}}{\delta f} = \left\langle \frac{\delta F}{\delta m}, p \right\rangle + \sum_{a=0}^A y^\prime_a(f) \frac{\delta F}{\delta s_a}.
\end{align}
We have suppressed the pullback $\pi^*$ of the cotangent bundle projection for notational simplicity. Here ${\delta F}/{\delta m}$ is a vector field on $M$, and ${\delta F}/{\delta s_a}$ is a function on $M$ for each $a$, so they all depend on $q$ only. Now we compute
\begin{align}\label{JAstep1}
\left\{ \tilde{F}, \tilde{G} \right\}_{KT} = & \int d^n q d^n p \ f \left\{ \frac{\delta \tilde{F}}{\delta f}, \frac{\delta \tilde{G}}{\delta f} \right\} \nonumber \\
= & \int d^nq d^n p \ f \left\{\left\langle \frac{\delta F}{\delta m}, p \right\rangle, \left\langle \frac{\delta G}{\delta m}, p \right\rangle  \right\} \nonumber \\
& \qquad + \sum_{a=1}^A  f \left\{y^\prime_a(f)\frac{\delta F}{\delta s_a}, \left\langle \frac{\delta G}{\delta m}, p \right\rangle  \right\} + \sum_{b=1}^A f \left\{\left\langle \frac{\delta F}{\delta m}, p \right\rangle, y^\prime_b(f)\frac{\delta G}{\delta s_b}  \right\} \nonumber \\
& \qquad + \sum_{a,b=1}^A f \left\{y^\prime_a(f)\frac{\delta F}{\delta s_a}, y^\prime_b(f)\frac{\delta G}{\delta s_b}  \right\}.
\end{align}
Using (\ref{qpPB1}), the term on the first line of (\ref{JAstep1}) is
\begin{align}
\int d^n q d^n p \ f \left\{\left\langle \frac{\delta F}{\delta m}, p \right\rangle, \left\langle \frac{\delta G}{\delta m}, p \right\rangle  \right\} 
=& - \int d^n q d^n p \ f \left \langle \left[ \frac{\delta F}{\delta m}, \frac{\delta G}{\delta m}\right], p\right\rangle,
\nonumber \\
=& -\left\langle m[f],  \left[ \frac{\delta F}{\delta m}, \frac{\delta G}{\delta m}\right] \right\rangle_{\mathrm{Vect}(M)}.
\end{align}
The second line of (\ref{JAstep1}) contains two similar groups of terms. We compute one of them as follows:
\begin{align}
&\int d^n q d^n p \ f \left\{y^\prime_a(f)\frac{\delta F}{\delta s_a}, \left\langle \frac{\delta G}{\delta m}, p \right\rangle  \right\}& \nonumber \\
= & \int d^n q d^n p \ \left(
f y_a^\prime(f) \left\{ \frac{\delta F}{\delta s_a}, \left\langle \frac{\delta G}{\delta m}, p \right\rangle\right\}
+ f y_a^{\prime\prime}(f) \frac{\delta F}{\delta s_a}\left\{ f, \left\langle \frac{\delta G}{\delta m}, p \right\rangle \right\}
\right), &\nonumber \\
= & \int d^n q d^n p  \left( 
\left(y_a(f) + ay_{a-1}(f) \right) \left\{ \frac{\delta F}{\delta s_a}, \left\langle \frac{\delta G}{\delta m}, p \right\rangle\right\}
+ \frac{\delta F}{\delta s_a}\left\{ a y_{a-1}(f), \left\langle \frac{\delta G}{\delta m}, p \right\rangle \right\}
\right), &\nonumber \\
 = & \int d^n q d^n p  \left( 
\left(y_a(f) + ay_{a-1}(f) \right) \left\{ \frac{\delta F}{\delta s_a}, \left\langle \frac{\delta G}{\delta m}, p \right\rangle\right\}
-  a y_{a-1}(f)\left\{\frac{\delta F}{\delta s_a}, \left\langle \frac{\delta G}{\delta m}, p \right\rangle \right\}
\right), &\nonumber \\
= & \int d^n q d^n p \ y_a(f) \left( \mathcal{L}_{\frac{\delta G}{\delta m}} \frac{\delta F}{\delta s_a} \right) = \left\langle s_a[f], \mathcal{L}_{\frac{\delta G}{\delta m}} \frac{\delta F}{\delta s_a} \right\rangle_{C^\infty (M)}. &
\end{align}
We have used the identity (\ref{ydiff}) as well as the recurrence relations (\ref{yarecur1},\ref{yarecur2}) of the functions $y_a(x) = x(\log x)^a$ to manipulate terms involving $y_a(f)$ and its derivatives in and out of the canonical Poisson bracket. The terms on the last line of (\ref{JAstep1}) vanish, because
\begin{align}\label{JAlaststep}
& \int d^n q d^n p \ f \left\{y^\prime_a(f)\frac{\delta F}{\delta s_a}, y^\prime_b(f)\frac{\delta G}{\delta s_b}  \right\} &\nonumber \\
 = & \int d^n q d^n p \ f \Bigg(
y^\prime_a(f)y^\prime_b(f)\left\{\frac{\delta F}{\delta s_a} , \frac{\delta G}{\delta s_b} \right\}
+ y^\prime_a(f)\frac{\delta G}{\delta s_b} \left\{ \frac{\delta F}{\delta s_a},y^\prime_b(f) \right\} &\nonumber \\
& \qquad \qquad \qquad \qquad
+ \frac{\delta F}{\delta s_a} y^\prime_b(f)\left\{y^\prime_a(f) , \frac{\delta G}{\delta s_b}\right\}
+ \frac{\delta F}{\delta s_a}\frac{\delta G}{\delta s_b}\left\{y^\prime_a(f) ,y^\prime_b(f) \right\}
\Bigg) &\nonumber \\
= & \int d^n q d^n p \ f \left(y^\prime_a(f)\frac{\delta G}{\delta s_b} \left\{ \frac{\delta F}{\delta s_a},y^\prime_b(f) \right\}+ \frac{\delta F}{\delta s_a} y^\prime_b(f)\left\{y^\prime_a(f) , \frac{\delta G}{\delta s_b}\right\}  \right) &\nonumber \\
= & \int d^n q d^n p \ \left(
f y^\prime_a(f)y^{\prime\prime}_b(f)  \frac{\delta G}{\delta s_b}\left\{ \frac{\delta F}{\delta s_a},f \right\}
+ f y^{\prime\prime}_a(f)y^\prime_b(f) \frac{\delta F}{\delta s_a}\left\{f , \frac{\delta G}{\delta s_b}\right\} 
\right) &\nonumber \\
= & \int d^n q d^n p \ \left(
\frac{\delta G}{\delta s_b}\left\{ \frac{\delta F}{\delta s_a},I_{ab}(f) \right\}
+  \frac{\delta F}{\delta s_a}\left\{I_{ba}(f) , \frac{\delta G}{\delta s_b}\right\}
\right) &\nonumber \\
= & -\int d^n q d^n p \ \left( I_{ab}(f) + I_{ba}(f)\right) \left\{\frac{\delta F}{\delta s_a} , \frac{\delta G}{\delta s_b} \right\} = 0.&
\end{align}
We have $ \{{\delta F}/{\delta s_a} , {\delta G}/{\delta s_b} \} = 0$ because both arguments are functions of $q$ only. The function $I_{ab}(x)$ that we have introduced on the second-to-last line of (\ref{JAlaststep}) is an antiderivative to $xy^\prime_a(x)y^{\prime\prime}_b(x)$:
\begin{align} \label{Iabdef}
I_{ab}(x) = \int_0^x dw \ w y^\prime_a(w)y^{\prime\prime}_b(w).
\end{align} 
We can express $xy^\prime_a(x)y^{\prime\prime}_b(x)$ as a polynomial in $\log x$, so the integral (\ref{Iabdef}) converges and $I_{ab}(x)$ is a linear combination of $x(\log x)^c$ for suitable integers $c$.

Putting everything together, we have
\begin{align}
\left\{ \tilde{F}, \tilde{G} \right\}_{KT} = &-\left\langle m[f],  \left[ \frac{\delta F}{\delta m}, \frac{\delta G}{\delta m}\right] \right\rangle_{\mathrm{Vect}(M)} \nonumber \\
& \qquad - \sum_{a=0}^A\left\langle s_a[f], \mathcal{L}_{\frac{\delta F}{\delta m}} \frac{\delta G}{\delta s_a} - \mathcal{L}_{\frac{\delta G}{\delta m}} \frac{\delta F}{\delta s_a} \right\rangle_{C^\infty (M)},
\end{align}
which is precisely the ($-$)-Lie--Poisson bracket on $\mathfrak{s}^*_A$ (\ref{sA_LP}), evaluated at the image of the Poisson map $\mathcal{J}_A[f] = (m[f],(s_0[f],\ldots,s_A[f]))$.
\end{proof}
The Poisson manifolds $C^\infty(T^*M)_+$ and $\mathfrak{s}^*_A$ are linear, in the sense that they are open subsets of vector spaces and that the Poisson brackets are linear. However, the Poisson map $\mathcal{J}_A$ is genuinely non-linear.
It is not induced by a Lie algebra action of $\mathfrak{s}_A$ on $T^*M$ by Hamiltonian vector fields for $A \geq 1$. 
The identity $\int dV \ \{f,g\} = 0$ is not needed for the $A=0$ case but is needed for the $A \geq 1$ cases. 

Nonetheless, the non-linearities in $\mathcal{J}_A$ are not completely arbitrary. They arise naturally from considering the \emph{generalised eigenvalue-eigenvector problem} \cite{Artin2014book} of the Euler operator $x(d/dx)$. For each $a=0,\ldots,A$, the functions $y_0(x),\ldots,y_a(x)$ form an ordered basis for the generalised eigenspace $V_1^a = \ker (x(d/dx)-1)^a$. The operator $x(d/dx)$ preserves the chain of inclusions $V_1^0\subset V_1^1 \subset \ldots \subset V_1^A$; in particular, $(x(d/dx)-1)V_1^a \subset V_1^{a-1}$. The non-linear components of $\mathcal{J}_A$ correspond to the generalised eigenvectors $y_1(x),\ldots,y_A(x)$, which is a natural extension of the ordinary eigenvector $y_0(x)$ that gives a linear component.

\begin{remark}
The true dimensionless expression for $\log f$ is $\log(h^n f)$, where $h$ is an arbitrary \enquote{Planck's constant} that defines an action ($[qp]$) scale. 
Changes in the reference action scale leads to transformations of the form $\log x \mapsto \log x + \Lambda$ in the definitions (\ref{yadef}) of $y_a(x)$.
Let us also allow for a different change of scale for each occurrence of $\log x$ in (\ref{yadef}), for example $x(\log x)^2$ to $x(\log x + \Lambda_1)(\log x + \Lambda_2)$. 
The general form of such a transformation is 
\begin{align}\label{yaLambda}
\tilde{y}_a = y_a(x) + \sum_{b<a} y_b(x) \Lambda_{ba},
\end{align}
where $\Lambda_{ba}$ is a strictly upper-triangular $(A+1)\times(A+1)$ matrix of constants, i.e. $\Lambda_{ba} = 0$ if $b \geq a$. The transformation (\ref{yaLambda}) includes (\ref{yaCa-b}) as a special case. The induced transformation $s_a[f] \mapsto \tilde{s}_a[f]$ on the generalised entropy densities is
\begin{align}\label{saLambda}
\tilde{s}_a = s_a + \sum_{b<a} s_b \Lambda_{ba},
\end{align}
which is a Poisson automorphism of $\mathfrak{s}^*_A$ i.e. an invertible Poisson map from $\mathfrak{s}^*_A$ onto itself. The transformation (\ref{saLambda}) adds constant multiples of the lower generalised entropy densities to the higher generalised entropy densities. The $a=1$ case of (\ref{saLambda}) reflects that the specific entropy $\eta_1 = s_1/s_0$ is only uniquely specified up to a global constant in classical thermodynamics.
\end{remark}

\begin{remark}
%Any Poisson map $\mathcal{J}: \mathfrak{M} \rightarrow \mathfrak{g}^*$, where $\mathfrak{g}^*$ is the dual of a Lie algebra $\mathfrak{g}$ equipped with the ($\pm$)-Lie--Poisson bracket, leads to an (anti-)homomorphism of Lie algebras $\mathcal{J}^\dagger: \mathfrak{g} \rightarrow \mathcal{F}(\mathfrak{M})$ given by $\mathcal{J}^\dagger(X) = \langle J(z),X \rangle_\mathfrak{g}$ where $X\in \mathfrak{g}$ and $z \in \mathfrak{M}$. 
It is of interest whether $\mathcal{J}_A: C^\infty(T^*M)_+ \rightarrow \mathfrak{s}^*_A$ is the momentum map for a Hamiltonian group action on $C^\infty(T^*M)_+$ by the semidirect product Lie group $\mathrm{Diff}(M) \ltimes C^\infty(M,\mathbb{R}^{A+1})$, where $\mathrm{Diff}(M)$ is the group of diffeomorphisms on $M$. 
For $A=0$, it is known that $\mathcal{J}_0$ is a momentum map for a Hamiltonian $\mathrm{Diff}(M) \ltimes C^\infty(M)$-action on $C^\infty(T^*M)_+$, induced by cotangent lifts of $\varphi \in \mathrm{Diff}(M)$ to $T^*M$ and momentum shifts $(q,p) \mapsto (q,p-dh)$ for $h \in C^\infty(M)$, both acting on $T^*M$ as Hamiltonian symplectomorphisms \cite{Marsden83, GuilleminSternberg84book, Guillemin80}.
However, there is no corresponding group action when $A \geq 1$. Consider the restriction of the pullback $\mathcal{J}_A^*:\mathcal{F}(\mathfrak{s}_A^*)\rightarrow \mathcal{F}(C^\infty(T^*M)_+)$ to linear functionals on $\mathfrak{s}_A^*$ that arise from dual pairing with elements in $\mathfrak{s}_A$. We denote the restriction by $\mathcal{J}_A^\dagger: \mathfrak{s}_A \rightarrow \mathcal{F}(C^\infty(T^*M)_+)$, given explicitly as
\begin{align}\label{Jdaggerdef}
\mathcal{J}_A^\dagger\left(u, (g_0,\ldots,g_A) \right) [f] = \left\langle \mathcal{J}_A[f], \left(u, (g_0,\ldots,g_A) \right) \right\rangle_{\mathfrak{s}_A},
\end{align}
for $(u, (g_0,\ldots,g_A) ) \in \mathfrak{s}_A$. The map $\mathcal{J}_A^\dagger$ (\ref{Jdaggerdef}) is an anti-homomorphism of Lie algebras because $\mathcal{J}_A$ is a Poisson map to $\mathfrak{s}_A^*$, which is equipped with the ($-$)-Lie--Poisson bracket. This leads to an $\mathfrak{s}_A$-action on $C^\infty(T^*M)_+$ by the Hamiltonian vector fields of $\mathcal{J}_A^\dagger(u, (g_0,\ldots,g_A) )[f]$, whose flow is
\begin{align}\label{Jdaggerflow}
\frac{\partial f}{\partial t} = \left\{\langle  u, p\rangle, f  \right\} + \sum_{a=0}^A \left\{   g_a, f(\log f)^a \right\},
\end{align}
where again we have suppressed the pullback $\pi^*$ of the cotangent bundle projection.

For $A \geq 1$, there are problematic flows from the abelian subalgebra $C^\infty(M,\mathbb{R}^{A+1})\subset \mathfrak{s}_A$ that cannot be integrated to one-parameter groups, which prevents us from integrating the Lie algebra action (\ref{Jdaggerflow}) to a group action, so $\mathcal{J}_A$ cannot be the momentum map of a $\mathrm{Diff}(M) \ltimes C^\infty(M,\mathbb{R}^{A+1})$-action.
We illustrate this with the example $M = S^1$ and $A = 1$. Let $(q,p)$ be coordinates on $T^*S^1 \simeq S^1 \times \mathbb{R}$. The Hamiltonian flow by $\mathcal{J}_1^\dagger (0, (0,g_1(q)))[f]$ on $ C^\infty(T^*S^1)_+$ is then
\begin{align}\label{JdaggerBurgers}
\frac{\partial f}{\partial t} - g_1^\prime(q)\frac{\partial f \log f}{\partial p} = 0, \quad \text{or} \quad \frac{\partial \Phi}{\partial t} - \left(g_1^\prime(q)+g_1^\prime(q)\Phi \right)\frac{\partial \Phi}{\partial p} = 0,
\end{align}
for $\Phi = \log f$. The basepoint $q\in S^1$ only enters parametrically. Fixing some $q$ at which $g_1^\prime(q) \neq 0$, the $\Phi$-equation in (\ref{JdaggerBurgers}) restricted to the fibre $T^*_q S^1$ can be rewritten as an inviscid Burgers' equation under an affine change of variables. A smooth solution to (\ref{JdaggerBurgers}) with smooth initial data $\Phi(t=0)$ breaks down at finite time
\begin{align}
\tau_b = \left(\mathrm{max} \left\{ g_1^\prime(q) \frac{\partial \Phi(q,p,t=0)}{\partial p} : {p\in T^*_qS^1} \right\} \right)^{-1}
\end{align}
if $g_1^\prime \partial \Phi(t=0) /\partial p > 0$, or equivalently if $g_1^\prime \partial f(t=0)/\partial p > 0$, for some $p$ \cite{Whitham74book, LeVeque92book}. 
If $f$ is smooth, non-constant and $f\rightarrow 0$ as $\lvert p \rvert \rightarrow \infty$, then $\partial f/\partial p$ must take both signs at different $p$, since $\int_{-\infty}^\infty dp \ \partial f/\partial p = 0$, so the breakdown of smooth solutions is unavoidable.
Thus the Hamiltonian flow (\ref{JdaggerBurgers}) by $\mathcal{J}_1^\dagger (0, (0,g_1(q)))$ on $ C^\infty(T^*S^1)_+$ only exists for finite time near each $f$; moreover, the existence time can be made indefinitely short by introducing large $p$-gradients in $f$.

On general manifolds $M$ and $A\geq 1$, the flow of $\mathcal{J}_A^\dagger (0, (0,g_1(q), 0,\ldots ,0))$ on $ C^\infty(T^*M)_+$ is of the form (\ref{Jdaggerflow}) along the direction spanned by $dg_1(q)$ on the fibre $T^*_qM$, with the other directions entering only parametrically. So we encounter the same problem and conclude that the Hamiltonian Lie algebra action (\ref{Jdaggerflow}) cannot arise from a group action for $A \geq 1$. The map $\mathcal{J}_A$ corresponds to a Hamiltonian Lie algebra action, but it is not a momentum map for a Hamiltonian group action.
\end{remark}
 
\subsection{A one-parameter family of Poisson maps from $C^\infty(T^*M)_+$ to $\mathfrak{s}^*_0$}\label{subsec:Jeps}

In this section we construct the Poisson map $\mathcal{J}_A$ (\ref{Jdef}) in a different way, by considering it as a Taylor series truncation of a one-parameter family of Poisson maps $\mathsf{J}_{\xi}:C^\infty(T^*M)_+ \rightarrow \mathfrak{s}^*_0$. These maps are defined by the dual pairing
\begin{align}\label{Jepsdef}
\left\langle \mathsf{J}_{\xi}[f], (u,g) \right \rangle_{\mathfrak{s}_0} = \int_{T^*M} dV \ \left( \langle u,p\rangle f + f^{1+\xi} g \right).
\end{align}
In components, we have $\mathsf{J}_{\xi}[f] = (m_i[f] dq^i \otimes d^n q,\rho_{\xi}[f] d^n q)$, where
\begin{align}\label{Jepsdefcpts}
m_i[f]  =  \int d^n p \ p_i f, \qquad \rho_{\xi}[f]  =  \int d^n p \ f^{1+\xi}.
\end{align}
The quantity $\rho_{\xi}$ is the spatial density of the so-called \emph{Tsallis entropy} \cite{Tsallis88} up to a multiplicative constant. For sufficiently regular and decaying $f$, the integrals in the definition of $\mathsf{J}_{\xi}[f]$ are finite for all $\xi$ in some interval $I$ containing $0$. Note that $\rho_{\xi}$ is not a fractional density despite being a $p$-integral of $f^{1+\xi}$ -- in terms of the density $\psi = f dV$, the expression to be integrated is $\psi (\psi/dV)^\xi$.
\begin{prop}\label{Jepsprop}
The one-parameter family of maps $\mathsf{J}_{\xi}:C^\infty(T^*M)_+ \rightarrow \mathfrak{s}^*_0$ (\ref{Jepsdef}) is a Poisson map for each $\xi \in I$, where $I$ is some appropriately chosen interval on $\mathbb{R}$ containing $0$.
\end{prop}
The proof can be found in appendix \textbf{\ref{App:proofs}} and is very similar to that of Proposition \ref{JPoissonprop}.
The function $x^{1+\xi}$, which is the form of the non-linearity in the map $\mathsf{J}_{\xi}$, is precisely the eigenfunction of the Euler differential operator $x(d/dx)$ with eigenvalue $1 + \xi$. The non-linearities in $\mathcal{J}_A$ (see section \textbf{\ref{subsec:JA}}) can thus be seen as an infinitesimal version of those in $\mathsf{J}_{\xi}$ as $\xi \rightarrow 0$. More precisely, the maps $\mathcal{J}_A$ can be obtained from expanding $\mathsf{J}_{\xi}$ in a Taylor series about $\xi = 0$ and then truncating. To see this, consider the Taylor expansion
\begin{align}\label{f1pxigen}
f^{1+\xi} = f \exp\left( \xi \log f \right) = \sum_{a=0}^\infty \frac{\xi^a}{a!} f \left(\log f \right)^a .
\end{align}
Thus $f^{1+\xi}$ is a generating function for the $y_a(f)$ defined in (\ref{yadef}). Integrating (\ref{f1pxigen}) over $p$ gives
\begin{align}
\rho_{\xi}[f] = \sum_{a=0}^\infty \frac{\xi^a}{a!} s_a[f],
\end{align}
so the truncated Taylor expansion of $\mathsf{J}_{\xi}[f]$ at $\xi=0$ of order $A$ gives the components of $\mathcal{J}_A[f]$ (\ref{Jdefcpts}). This argument shows that we can obtain $\mathcal{J}_A$ from the Poisson map $\mathsf{J}_{\xi}$, but it does not directly show that $\mathcal{J}_A$ is a Poisson map.
In the following we present a more geometrical approach from which $\mathcal{J}_A$ is obtained as the composition of Poisson maps, hence manifestly a Poisson map.
 
Consider the \emph{path space $PC^\infty(M)$ of one-parameter families of functions on $M$}. The elements of $PC^\infty(M)$ are maps $g_\xi: I \rightarrow C^\infty(M)$. The Lie algebra of vector fields $\mathrm{Vect}(M)$ has a representation on $PC^\infty(M)$ by acting on the image. More explicitly, the action of a vector field $u$ on $g_\xi$ is $g_\xi \mapsto \mathcal{L}_u g_\xi$. This gives rise to a semidirect product Lie algebra $\mathrm{Vect}(M) \ltimes PC^\infty(M)$.

We take the smooth dual space to this semidirect product Lie algebra to be $\mathrm{Vect}(M)^* \times P\Omega^n(M)$, where $P\Omega^n(M)$ is the \emph{path space $P\Omega^n(M)$ of top-degree forms (or densities) on $M$}, i.e. the space of one-parameter families of $n$-forms $\rho_\xi: I \rightarrow \Omega^n(M)$. The dual pairing between $P\Omega^n(M)$ and $PC^\infty(M)$ is given by integration over $I \times M$, where $I$ is equipped with the standard measure $d\xi$. The one-parameter family of Poisson maps $\mathsf{J}_{\xi}$ in Proposition \ref{Jepsprop} can be considered as a single map $\mathsf{J}:C^\infty(T^*M)_+ \rightarrow  \mathrm{Vect}(M)^* \times P\Omega^n(M)$, defined by
\begin{align}\label{Jbigdef}
\left\langle \mathsf{J}[f], (u,g_\xi) \right\rangle =  \int_{T^*M} dV \ \left( \langle u,p\rangle f + \int_I d\xi \ \left(f^{1+\xi} g_\xi\right) \right) .
\end{align}

\begin{prop}\label{Jbigrop}
The map $\mathsf{J}:C^\infty(T^*M)_+ \rightarrow \mathrm{Vect}(M)^* \times P\Omega^n(M)$ (\ref{Jbigdef}) is a Poisson map.
\end{prop}
The proof, which is a small variation of that for Proposition \ref{Jepsprop}, can be found in appendix \textbf{\ref{App:proofs}}.
Having constructed the Poisson map $\mathsf{J}$, we proceed to construct the second Poisson map that would give $\mathcal{J}_A$ when composed to $\mathsf{J}$. This second map arises from the theory of \emph{jet spaces}, which codify Taylor expansions in a geometrical manner \cite{Omohundro86book, Kolar93book}. Consider the sequence of projections
\begin{align}\label{jettower}
& J_0^0(I, \Omega^n(M)) \leftarrow J_0^1(I, \Omega^n(M)) \leftarrow \ldots  \leftarrow J_0^\infty(I, \Omega^n(M))  \leftarrow P \Omega^n(M),
\end{align} 
where $J_0^A(I,\Omega^n(M))$ is the space of \emph{$A$-jets of maps $\rho_\xi: I\rightarrow \Omega^n(M)$ at $\xi = 0$}, defined as the quotient space of $P \Omega^n(M)$ under the following equivalence relation: $\rho_\xi \sim \tilde{\rho}_\xi$ if and only if the Taylor expansion of $\rho_\xi$ and $\tilde{\rho}_\xi$ at $\xi=0$ agree up to order $A$. The \emph{infinite jet space} $J_0^\infty(I, \Omega^n(M))$ is defined similarly, with the equivalent relation being the agreement of Taylor expansions up to arbitrary order. Each projection $J_0^{A} \rightarrow J_0^{B}$ for $A>B$ in (\ref{jettower}) corresponds to a further truncation of the Taylor expansion of $\rho_\xi$.

To simplify notation, we use the shorthand $\mathrm{V}$ for $\mathrm{Vect}(M)$ and omit the arguments in the jet spaces for the rest of the section. We would like to show that the sequence of projections 
\begin{align}\label{jettower2}
\mathrm{V}^* \times J_0^0 \leftarrow \mathrm{V}^* \times J_0^1 \leftarrow \ldots \leftarrow \mathrm{V}^* \times  J_0^\infty
 \leftarrow \mathrm{V}^* \times P \Omega^n(M) 
\end{align}
are Poisson maps between Poisson manifolds, so that composing $\mathsf{J}$ with these projections corresponds to truncating the Taylor series of $\rho_\xi$, hence giving the Poisson map $\mathcal{J}_A$. This can be achieved by considering the spaces in (\ref{jettower2}) as smooth duals of Lie algebras and the maps as adjoints of injective Lie algebra homomorphisms. Consider the space of $C^\infty(M)$-valued distributions on $I$, which we denote by $D(I,C^\infty(M))$ or simply $D$ for brevity. The Lie algebra $\mathrm{V}$ acts on $D$ by Lie differentiation on its values: $g_\xi \mapsto \mathcal{L}_u g_\xi$, where $u$ is a vector field and $g_\xi$ is a $C^\infty(M)$-valued distribution. The subspace $D^A_0$ of distributions that are supported at $0$ and are of order at most $A$ is invariant under this action. An element of $D^A_0$ can always be written in the form 
\begin{align}\label{gepsexample}
g_\xi = \sum_{a=0}^{A} g_a (-1)^a \frac{d^a}{d \xi^a} \delta(\xi), \quad \text{where } g_a \in C^\infty(M).
\end{align}
The subspace $D^\infty_0 = \cup_{a=0}^\infty D^a_0$ is also invariant. The inclusions of these invariant subspaces induces a sequence of injective homomorphisms between semidirect product Lie algebras:
\begin{align}\label{jettower3}
\mathrm{V} \ltimes D_0^0 \xhookrightarrow{} \mathrm{V} \ltimes D_0^1 \xhookrightarrow{} \ldots \xhookrightarrow{} \mathrm{V} \ltimes  D_0^\infty  \xhookrightarrow{} \mathrm{V} \ltimes D.
\end{align}
We recover the sequence (\ref{jettower2}) from dualising the sequence of injective homomorphisms between Lie algebras (\ref{jettower3}). This endows each space in (\ref{jettower2}) with a Lie--Poisson structure, and makes each map in (\ref{jettower2}) a Poisson map by being the adjoint of a Lie algebra homomorphism.

Now we can explain the connection between $\mathsf{J}$ and $\mathcal{J}_A$. Consider the isomorphism $D^A_0 \simeq C^\infty(M, \mathbb{R}^{A+1})$ given by the map $g_\xi \mapsto (g_0,\ldots, g_A)$, where the $g_a$ are as defined in (\ref{gepsexample}). The action of $\mathrm{V}$ on $D^A_0$ corresponds to componentwise Lie differentiation on $C^\infty(M, \mathbb{R}^{A+1})$ under this isomorphism. We can thus identify $\mathrm{V} \ltimes D_0^A$ with $\mathfrak{s}_A$ as semidirect product Lie algebras (see section \textbf{\ref{subsec:LPforfluids}}), and hence identify $\mathrm{V}^* \times J_0^A$ with $\mathfrak{s}^*_A$ as Poisson manifolds. The Poisson map $\mathcal{J}_A$ (\ref{Jdef}) is then recovered from composing $\mathsf{J}$ (\ref{Jbigdef}) with the sequence of projections (\ref{jettower2}). Indeed, the pairing of $\rho_{\xi}[f]$ with $g_\xi$ in (\ref{gepsexample}) is
\begin{align}
\left\langle \rho_\xi[f], g_\xi \right\rangle_{PC^\infty(M)} = &\int_I d\xi \ \left\langle  \rho_\xi[f], g_\xi\right\rangle_{C^\infty(M)}, \nonumber \\
= & \int_I d\xi \ \sum_{a=0}^A (-1)^a \frac{d^a \delta(\xi)}{d \xi^a}  \left\langle \rho_\xi[f], g_a \right\rangle_{C^\infty(M)}, \nonumber \\ 
 = &\sum_{a=0}^{A}\left\langle  s_a[f], g_a\right\rangle_{C^\infty(M)}.
\end{align}
The Poisson map $\mathsf{J}$ can be truncated at or expanded around other values of $\xi$ with a similar reasoning as above, by considering distributions on $I$ that are supported at a finite number of points $\xi_1,\ldots,\xi_k$, with prescribed (finite) order $A_j$ at each point $\xi_j$. This creates a Poisson map from $C^\infty(T^*M)_{+}$ to $\mathfrak{s}^*_A$, where $A+1 = \sum_{j=1}^k (A_j+1)$ is the number of $C^\infty(M)$-degrees of freedom.

In particular, if we consider distributions on $I$ of order zero supported at the \emph{non-negative integers $\xi = 0,\ldots A$}, we obtain a Poisson map $\mathcal{J}_A^{pol}: C^\infty(T^*M)_{+} \rightarrow\mathfrak{s}^*_A$, given explicitly as
\begin{align}\label{Jpoldef}
\mathcal{J}_A^{pol}[f] = \left(m_i [f] dq^i \otimes d^n q,(\rho^{pol}_0[f] d^n q,\ldots, \rho^{pol}_A[f] d^n q)\right),
\end{align}
where
\begin{align}
m_i[f] =  \int d^n p \ p_i f, \qquad \rho^{pol}_{a}[f]  =  \int d^n p \ f^{1+a}.
\end{align}
The condition $f>0$ can be relaxed for the map $\mathcal{J}_A^{pol}$, since positive integer powers of $f$ are uniquely defined even when $f$ is not positive. Thus we can extend the domain of $\mathcal{J}_A^{pol}$ to all of $C^\infty(T^*M)$. Moreover, since $\mathcal{J}_A^{pol}$ can be defined using only integer powers of $f$, the pullback $(\mathcal{J}_A^{pol})^* F$ of a \emph{polynomial} functional $F$ on $\mathfrak{s}_A^*$ is a polynomial functional on $C^\infty(T^*M)$.

\section{Fluid dynamics as a Hamiltonian approximation to kinetic theory}\label{sec:Hamiltonian}
Following the discussion on noncanonical Hamiltonian systems in section \textbf{\ref{sec:physics}}, to obtain a reduced Hamiltonian system from a more primitive one, we first seek a Poisson map between the underlying Poisson manifolds of the two systems, then attempt to show that the primitive Hamiltonian functional factors through the Poisson map. We have produced the relevant Poisson map $\mathcal{J}_1$ in section \textbf{\ref{sec:Poissonmap}}. In this section, we analyse the relevant Hamiltonian functionals in $1$-particle kinetic theory and hydrodynamics with non-constant entropy. We will find that this reduction cannot be carried out exactly because the kinetic-theory Hamiltonian does not factor through any of the Poisson maps we have constructed in section \textbf{\ref{sec:Poissonmap}}. However, we can obtain an approximate Hamiltonian that factors through $\mathcal{J}_1$ by ignoring the contribution from the entropy of the distribution relative to its local Maxwellian.
In the following, we take $M = \mathbb{T}^n$ or $\mathbb{R}^n$ for simplicity and suppress the volume element $d^n q$ from the notation. The discussion in this section (except \textbf{\ref{subsec:globalM}}) can be generalised to arbitrary Riemannian manifolds without difficulty. We will also revert to the more conventional notation for the mass density $\rho = s_0$ and the entropy density $s = s_1$, but keep the mathematician's sign convention for the entropy.

Consider a $1$-particle distribution function $f=f(q,p)$ that describes the state of a monatomic ideal gas. We suppress the time dependence on $f$ throughout, since $t$ only enters parametrically. We define the quantities
\begin{align}
\rho[f] & = \int d^n p \ f(q,p), & \\
m[f] & = \int d^n p \ p f(q,p), & u[f] = \frac{m[f]}{\rho[f]}, \\
s[f] & = \int d^n p \ f(q,p)\log f(q,p), \label{sdef}\\
\theta[f] & = \frac{1}{n\rho} \int d^n p \ \lvert p-u\rvert^2 f(q,p), 
\end{align}
Here $(m[f],\rho[f],s[f])$ are the components of the Poisson map $\mathcal{J}_1: C^\infty(T^*M)_+ \rightarrow \mathfrak{s}^*_1$, which we interpret as the momentum density, mass density and entropy density respectively, $u[f]$ is the mean velocity, and $\theta[f]$ is the \emph{kinetic theory temperature}. These quantities are all \emph{hydrodynamic variables} or \emph{fluid moments}, i.e. fields depending only on $q$ obtained from certain $p$-integrals involving $f$. In terms of these fluid moments, the spatial density of the energy of the ideal gas can be split into a mean part $\rho u^2/2$ corresponding to the kinetic energy of the mean motion of the particles, and a thermal part $n\rho\theta/2$ corresponding to the kinetic energy of the motion of the particles relative to that mean. The splitting reads
\begin{align}\label{HKT_idealgas}
H_{KT}[f] = \int d^n q d^n p \ \frac{\lvert p \rvert^2}{2} f = \int d^n q \ \left(\frac{\lvert m[f]\rvert^2}{2\rho[f]} + \frac{n}{2}\rho[f]\theta[f]\right).
\end{align}

A similar splitting also appears in the hydrodynamic description of the gas. The total energy density of the fluid (\ref{Hfluids0}) is split into a kinetic part $m^2/2\rho$ and a thermal part $\rho U(\rho,s)$, where $U(\rho, s)$ is the internal energy per unit mass of the fluid as determined by a thermodynamic equation of state. For a monatomic ideal gas, we have $U = (n/2)T$, where $T(\rho,s)$ is the \emph{thermodynamic temperature} of the gas, defined implicitly by
\begin{align}\label{EOS_idealgas_raw}
\frac{s}{\rho} = \log\left( \rho T^{-{n}/{2}} \right) + C.
\end{align}
Here $s /\rho$ is the local specific entropy, and $C$ is an arbitrary constant. The Hamiltonian functional that generates the compressible Euler equations for a monatomic ideal gas is
\begin{align}\label{Hfluids_idealgas}
H_{fluids}[m,\rho,s] = \int d^n q \ \left(\frac{\lvert m\rvert^2}{2\rho} + \frac{n}{2}\rho T(\rho,s)\right),
\end{align}
where $T(\rho,s)$ is obtained from rearranging (\ref{EOS_idealgas_raw}).

Comparing (\ref{HKT_idealgas}) and (\ref{Hfluids_idealgas}), both Hamiltonians contain a kinetic part $m^2/2\rho$ and a thermal part. The kinetic parts agree under the Poisson map $\mathcal{J}_1$, so it is tempting to hope that the kinetic-theory temperature $\theta[f]$ can be described in terms of $(m[f],\rho[f],s[f])$, preferably through the equation of state $T(\rho,s)$. If this holds, $H_{KT}$ would be the pullback of $H_{fluids}$ through the Poisson map $\mathcal{J}_1$, so that the reduction of noncanonical Hamiltonian systems described in section \textbf{\ref{sec:physics}} can be carried out exactly. 

Unfortunately, this exact reduction is not possible. The extent to which it fails is quantified by the deviation of the distribution function $f$ from its associated \emph{local Maxwellian} $f_m$, which is defined as
\begin{align}\label{localMdef}
f_m(q,p) & = \frac{\rho}{(2\pi\theta)^{n/2}}\exp\left( -\frac{\lvert p-u\rvert^2}{2\theta} \right).
\end{align}
The local Maxwellian $f_m$ is so constructed to have the same mass density, mean velocity and kinetic-theory temperature as $f$. The dependence of $f_m$ on $f$ is highly non-linear. Substituting $f_m$ into (\ref{sdef}) and rearranging gives
\begin{align}\label{theta0}
\theta[f]  = \frac{\rho[f]^{2/n}}{2\pi\mathrm{e}}\exp\left( - \frac{2}{n} \frac{s[f_m]}{\rho[f]} \right),
\end{align}
where $\mathrm{e} = \exp(1)$. On the other hand, choosing $C = -(n/2)\log(2\pi\mathrm{e})$ in (\ref{EOS_idealgas_raw}) and rearranging gives
\begin{align}\label{EOS_idealgas}
T(\rho[f],s[f]) = \frac{\rho[f]^{2/n}}{2\pi \mathrm{e}}\exp\left(-\frac{2}{n} \frac{s[f]}{\rho[f]} \right).
\end{align}
The expressions (\ref{theta0}) and (\ref{EOS_idealgas}) are remarkably similar, with the only difference being that $s[f]$ appears in (\ref{theta0}) while $s[f_m]$ appears in (\ref{EOS_idealgas}). This difference can also be expressed in terms of the relative entropy density of $f$ against $f_m$.

We define the \emph{relative entropy density} of $f_1$ against $f_2$ as \cite{Cercignani88,SaintRaymond09book}
\begin{align}\label{relentloc}
r[f_1 \vert f_2] = \int d^n p \ \left( f_1 \log(f_1/f_2) + f_1 - f_2 \right),
\end{align}
where $f_1$ and $f_2$ are positive distribution functions. Here $r[f_1 \vert f_2]$ is a non-negative function of $q$, with equality at fixed $q_0$ if and only if $f_1(q_0,p) = f_2(q_0,p)$.
We can show that $s[f]-s[f_m] = r[f \vert f_m]$ using properties of the local Maxwellian \cite{SaintRaymond09book}, so
\begin{align}\label{theta2}
\theta[f] = T\left(\rho[f],s[f_m]\right) = T\left(\rho[f],s[f]\right) \exp\left( \frac{2}{n} \frac{r[f \vert f_m]}{\rho[f]} \right).
\end{align} 
From the non-negativity of $r[f \vert f_m]$ and (\ref{theta2}), we see that $\theta[f]\geq T(\rho[f],s[f])$ with equality if and only if $f=f_m$. Putting everything together, we have:
\begin{prop}\label{cannotreduceprop}
The kinetic-theory Hamiltonian $H_{KT}[f]$ from (\ref{HKT_idealgas}), defined on the space $C^\infty(T^*M)_{+}$ of positive $1$-particle distribution functions, is related to the fluid Hamiltonian $H_{fluids}[m,\rho,s]$ from (\ref{Hfluids_idealgas}), defined on the space $\mathfrak{s}^*_1$ of hydrodynamic variables for an ideal compressible fluid with non-constant entropy by
\begin{align}\label{cannotreduce}
H_{KT} = \mathcal{J}_1^*H_{fluids} + \Delta H,
\end{align}
where $\mathcal{J}_1:C^\infty(T^*M)_{+} \rightarrow \mathfrak{s}^*_1$ is the Poisson map described in (\ref{Jdef}) with $A=1$, and
\begin{align}\label{deltaHdef}
\Delta H[f] = \frac{n}{2}\int d^n q \ \rho[f] \theta[f] \left(1 - \exp\left(-\frac{2}{n} \frac{r[f \vert f_m]}{\rho[f]} \right) \right).
\end{align} 
This functional satisfies the inequality $\Delta H[f]\geq 0$, with equality if and only if the distribution function $f$ is a local Maxwellian everywhere, i.e. $f(q,p)=f_m(q,p)$.
\end{prop}

A similar result holds for the Vlasov--Poisson equation. The extra mean-field potential in the Hamiltonian functional $H_{VP}[f]$ (\ref{HVlasov}) depends on $\rho[f]$ only, so it also factors through the Poisson map $\mathcal{J}_1$. If we ignore the relative entropy of the distribution function against its local Maxwellian, the Vlasov--Poisson system reduces to the Euler--Poisson equations \cite{DietzSandor99, ZakharovKuznetsov97}, with non-constant entropy.
\begin{corollary}\label{cannotreducecor}
The Vlasov--Poisson Hamiltonian $H_{VP}[f]$ from (\ref{HVlasov}) can be written as
\begin{align}
H_{VP} = \mathcal{J}_1^*\left(H_{fluids}+H_{static}\right) + \Delta H,
\end{align}
where $H_{fluids}[m,\rho,s]$ is the fluid Hamiltonian (\ref{Hfluids_idealgas}) for an uncharged monatomic ideal gas, $\Delta H$ is as in (\ref{deltaHdef}), and $H_{static}[\rho]$ is the electrostatic energy of a charged fluid, defined by
\begin{align}
H_{static}[\rho] = \frac{{e}^2}{2}\int d^nq d^nq^{\prime} \rho(q)\rho(q^{\prime})G(q,q^{\prime}),
\end{align}
where $G(q,q^\prime)$ is the Green's function of the Laplacian.
\end{corollary}
An entirely analogous result holds for a self-gravitating system as well, since the potential energy in the self-gravitating Hamiltonian $H_{SG}[f]$ (\ref{HJeans}) is exactly that of $H_{VP}[f]$, except that $e^2$ is replaced with $-\mathrm{G}$. In this case ignoring the relative entropy reduces the Jeans equation to the Euler--Poisson equations under the same replacement. 

\subsection{Discussion}\label{subsec:Hdiscuss}
We cannot reduce the kinetic theory description for an ideal gas to a hydrodynamic description exactly, even if we include the entropy density using the Poisson map $\mathcal{J}_1$. 
Suppose that we are given an unknown distribution function $f$, and that we know the values of the integrals $m[f] = \int d^np \ p f$ and $\int d^n p \ y(f)$ for all real functions $y$. From this information alone, we cannot determine the kinetic-theory temperature $\theta[f]$ in general.
At a fixed $q$, if we rearrange $p$-space in a volume-preserving way, the values of the integrals $\int d^n p \ y(f)$ remain unchanged. However, there are examples of such volume-preserving rearrangements where $m[f]$ remains unchanged while $\theta[f]$ changes. 
Indeed, consider the one-dimensional example $f_c(p) = (B(p-c) + B(p+c))/2$, where $B(x)$ is a bump function with support in $[-1/2,1/2]$. For all $c \geq 1$, the integrals $\int dp \ y(f_c)$ for any given real function $y$ coincide, yet $\theta[f_c]$ can be made arbitrary large as $c$ increases. 
We conclude that $\theta[f]$, and hence $H_{KT}[f]$, cannot be determined from the image $\mathcal{J}_A[f]$ (\ref{Jdefcpts}) of $f$ under the Poisson map $\mathcal{J}_A$. (The same is true for the maps $\mathsf{J}_{\xi}$ and $\mathsf{J}$ considered in section \textbf{\ref{subsec:Jeps}}.) This is an instance of the more general \emph{moment closure problem} in kinetic theory, where an exact description of a kinetic system using a finite number of fluid moments is impossible in general. 

An equivalent Poisson map using kinetic-theory temperature instead of entropy cannot be constructed due to the lack of finite closure of fibrewise-polynomials in $p$ of degree $2$ or higher (see section \textbf{\ref{subsec:cotangent}}). While there is a Poisson map that includes the trace of the second moment $\int d^n p \ \lvert p \rvert^2 f$ \cite{Gibbons81, GibbonsHolmTronci08}, such a map necessarily includes $p$-moments of all orders, leading to a system with an infinite number of $C^\infty(M)$-degrees of freedom. If one attempts to form the Poisson bracket of two functionals that depend on the first $l$ moments only, moments of order up to $2l-1$ appear in resulting expression. A closed description that keeps the $l=2$ moment would then require all moments to be kept.

One way out of this dilemma is to impose a \emph{closure relation} between the higher and lower $p$-moments in the Poisson bracket \cite{PerinChandreMorrisonTassi15,Tassi15,Tassi16,Tassi17}. One defines a new Poisson bracket that operates on functionals that depend on the first $l$ $p$-moments only, by first forming their kinetic-theory Poisson bracket, then eliminating the moments of order greater than $l$ using the closure relation. The closure relation needs to be chosen in a way that retains the Jacobi identity. The map to the retained $l$ $p$-moments is not a Poisson map, but the kinetic-theory Hamiltonian factors through this map (see section \textbf{\ref{subsec:notmyreduction}}). This is an example of an \emph{approximate closure}, where one makes uncontrolled approximations in order to derive a reduced system. In this case, the uncontrolled approximation is the imposition of the closure relation on the kinetic-theory Poisson bracket.

In this paper we present yet another way to obtain an approximate closure, where we keep the respective Poisson brackets in kinetic theory and fluid dynamics unchanged, but make an uncontrolled approximation to the kinetic-theory Hamiltonian. Proposition \ref{cannotreduceprop} suggests that, if we neglect the term $\Delta H$ in the kinetic-theory Hamiltonian $H_{KT}$ (\ref{cannotreduce}) by ignoring the relative entropy density $r[f|f_m]$, the remaining non-linear functional is precisely the pullback of the fluid Hamiltonian $H_{fluids}$ through the Poisson map $\mathcal{J}_1$. For small deviations from a local Maxwellian $f = f_m(1+\epsilon h)$, the relative entropy density $r[f\vert f_m]$ and hence the term $\Delta H$ are formally $O(\epsilon^2)$, so $\mathcal{J}_1^*H_{fluids}$ can be thought of as a truncation of $H_{KT}$ in the limit $\epsilon \rightarrow 0$. This approximation gives a manifestly Hamiltonian derivation of the compressible Euler equations, endowing them with a Hamiltonian structure that is inherited from kinetic theory through $\mathcal{J}_1$. Proposition \ref{cannotreduceprop} is an improvement of the result in \cite{Marsden83}, where the Poisson map $\mathcal{J}_0$ is used, the entropy density $s$ is excluded, and the temperature term $\theta[f]$ is replaced by hand with an internal energy term $U[\rho]$.

This should be compared to the usual derivation of hydrodynamics from kinetic theory, where the \emph{(linearised) collisional Boltzmann equation} in the strongly collisional limit $\mathrm{Kn} \rightarrow 0$ is considered instead of collisionless Boltzmann equation. Using a multiple-scales expansion known as the \emph{Chapman--Enskog expansion}, one finds the compressible Euler equations at the zeroth order and the Navier--Stokes--Fourier equations at the first order \cite{ChapmanCowling70book, Cercignani90book}. The Chapman--Enskog expansion offers a formal derivation of hydrodynamics as an approximation to kinetic theory in the strongly collisional limit, but it does not explain why the compressible Euler equations are Hamiltonian. Conversely, Proposition \ref{cannotreduceprop} addresses the Hamiltonian structure of the compressible Euler equations, but has no bearing on their validity as an asymptotic limit of kinetic theory. 
Similarly, we can obtain a \enquote{naive fluid model} for an electrostatic plasma by ignoring $\Delta H$ in $H_{VP}[f]$, and likewise for a self-gravitating system. Corollary \ref{cannotreducecor} suggests that the naive fluid models obtained this way inherit the Hamiltonian structure of $1$-particle kinetic theory through the Poisson map $\mathcal{J}_1$, but it has no bearing on the physical validity of such models.

This type of derivation of reduced conservative models through uncontrolled approximations or constraints in a Lagrangian or Hamiltonian framework is not uncommon in fluid dynamics and mathematical physics. For example, the shallow water and Green--Naghdi equations can be derived by imposing constraints on columnar motion in the variational principle for general ideal compressible fluids \cite{MilesSalmon85}. In geophysical fluid dynamics, the small Rossby number approximation can be made in the variational principle to derive manifestly conservative equations for nearly geostrophic flow \cite{Salmon85}. In the theory of modulated waves, where waves propagate in a slowly varying media, the Whitham modulation equations can be derived by averaging the fast scales in the variational principle \cite{Whitham67,Whitham70,Whitham74book}. Similarly, in non-linear optics, the particle-like reduced dynamics of solitons in the discrete non-linear Schr\"{o}dinger equation can be obtained from making an ansatz in the variational principle \cite{MalomedWeinstein96, Kevrekidis09book}. In semiclassical mechanics, the Hamilton--Jacobi equation can be obtained from the Schr\"{o}dinger equation by first performing the Madelung transform and then taking the classical limit $\hbar \rightarrow 0$ in the Hamiltonian functional \cite{KhesinMisiolekModin20,KhesinMisiolekModin19}. Finally, the Vlasov--Poisson equation can be derived from the BBGKY hierarchy by substituting the ansatz $f_2 =f_1 f_1$ in the Hamiltonian functional \cite{MarsdenMorrisonWeinstein84}.

A thoughtful discussion on the nature of reduced models derived through ignoring formally small terms can be found in \cite{Salmon85}, pages 473--474. In general, one cannot expect the dynamics of such reduced systems to stay close to that of the primitive system after a finite amount of time, since the neglected terms act like small sources of error that accumulate over time. However, by making these approximations in a Hamiltonian way, the reduced system can be made to possess desirable \emph{structural properties}. For example, we can have exact conservation laws for the approximated forms of physical quantities, such as the conservation of potential vorticity in various dispersive shallow-water models and in nearly geostrophic flows \cite{MilesSalmon85,Salmon85}. Another example would be the promotion of adiabatic invariants of the primitive system into exact invariants of the reduced system, such as the wave action for a wave propagating in a slowly varying background \cite{Whitham67,Whitham70,Whitham74book}.

\subsection{The near global Maxwellian regime}\label{subsec:globalM}

In this section we show that the term $\Delta H$ in (\ref{cannotreduce}) stays small uniformly in time if the initial distribution function $f$ is sufficiently close to a \emph{global Maxwellian}, in a sense made precise by the relative entropy. This is not a proof that the solutions of the collisionless Boltzmann equation approximates the solutions of the compressible Euler equations.

In the following we work exclusively on $\mathbb{T}^n$. This choice is motivated by the need to define the quantities
\begin{align}
\label{rho0def}
\rho_0 & = \frac{1}{\mathrm{vol}\left(\mathbb{T}^n\right)}\int d^n q d^n p \ f(q,p), & \\
\label{u0def}
\rho_0 u_0 & = \frac{1}{\mathrm{vol}\left(\mathbb{T}^n\right)}\int d^n q d^n p \ p f(q,p), \\
\rho_0 \left(\lvert u_0\rvert^2 + n\theta_0 \right) & = \frac{1}{\mathrm{vol}\left(\mathbb{T}^n\right)}\int d^n q d^n p \ \lvert p\rvert^2 f(q,p), \label{theta0def} \\
f_M(q,p) & = \frac{\rho_0}{(2\pi\theta_0)^{n/2}}\exp\left( -\frac{\lvert p-u_0\rvert^2}{2\theta_0} \right).
\end{align}
Here $\rho_0$ is the mean mass density, $u_0$ is the mean velocity, and $\theta_0$ is the mean temperature. The function $f_M$ is the \emph{global Maxwellian} with the same mean mass density, mean velocity and mean temperature as $f$. The mean mass density is only defined on manifolds with finite volume, and there is no straightforward interpretation of the mean velocity on manifolds that are not domains in $\mathbb{R}^n$ or $\mathbb{T}^n$, so $\mathbb{T}^n$ is simplest choice of a connected manifold without boundary on which the global Maxwellian is unambiguously defined.

The mean mass density $\rho_0$ is a Casimir functional of the kinetic-theory Poisson bracket described in section \textbf{\ref{subsec:LPforKT}}. The mean velocity $u_0$ and mean temperature $\theta_0$ are constants of motion because they Poisson commute with $H_{KT}$ (\ref{HKT_idealgas}).

We define the total \emph{relative entropy} of $f_1$ against $f_2$, where $f_1$ and $f_2$ are (positive) distribution functions as the $q$-integral of the relative entropy density $r[f_1 \vert f_2]$ (\ref{relentloc}):
\begin{align}\label{relenttot}
R[f_1 \vert f_2] = \int d^n q d^n p \ \left( f_1 \log(f_1/f_2) + f_1 - f_2 \right).
\end{align}
This quantity is non-negative and is zero if and only if $f_1 = f_2$. The relative entropy $\mathcal{R}_{in}$ of $f$ against its associated global Maxwellian
\begin{align}\label{Hindef}
\mathcal{R}_{in} = R[f \vert f_M] = \int d^n q d^n p \ f\log f -  \int d^n q d^n p \ f\log f_M .
\end{align}
is a constant of motion, since the first term is a Casimir functional of the kinetic theory bracket, and the second term is a constant of motion. Now we state the main result of this section:
\begin{prop}\label{boundsprop}
Let $\Delta H[f]$ be as in (\ref{cannotreduce}). Suppose that $f$ evolves according to the collisionless Boltzmann equation. Then the following inequality holds:
\begin{align}\label{dH_inequality0}
\Delta H[f] \leq 2\theta_0 {\mathcal{R}_{in}},
\end{align}
where the right-hand side consists only of constants of motion.
\end{prop}

\begin{proof}
To avoid notational clutter, we suppress the $f$-dependence from functionals of $f$ except for the relative entropies. We would like to bound the term
\begin{align}\label{dH_inequality1}
\Delta H = & \frac{n}{2}\int d^n q  \ \rho \theta \left(1 - \exp\left(-\frac{2}{n}\frac{r[f\vert f_m]}{\rho}\right) \right),
\end{align}
which contains the product of the temperature with an entropy-like term. This can be done using the physical a priori estimates in \cite{SaintRaymond09book}, chapter 3.1. First note that
\begin{align}\label{Hinsplit}
\mathcal{R}_{in} = R[f\vert f_M] = R[f \vert f_m] + R[f_m \vert f_M].
\end{align} 
The first term is the relative entropy of $f$ against its local Maxwellian $f_m$, and the second term is the relative entropy of the local Maxwellian $f_m$ against the global Maxwellian $f_M$. Both quantities are manifestly non-negative. To simplify notation, we write $\mathcal{R}_{m}=R[f \vert f_m]$ and $\mathcal{R}_{M}=R[f_m \vert f_M]$.
\begin{enumerate}
\item First we study the entropy-like term $\eta = 1 - \exp\left(-2r[f\vert f_m]/n\rho\right)$. From the elementary inequality $0\leq 1 - \exp(-x) \leq \min (1,x)$ for $x\geq 0$, we have
\begin{align}\label{S_ineq1}
0\leq \eta\leq 1, \quad \text{and } \int d^nq \ \rho \eta \leq \frac{2}{n}\mathcal{R}_{m}.
\end{align}
We can combine these inequalities to obtain control over any power of $\eta$ in terms of the relative entropy:
\begin{align}\label{S_ineq2}
\int d^nq \ \rho \eta^a \leq \frac{2}{n}\mathcal{R}_{m} \quad \text{for all $a\in\mathbb{Z}^+$,}
\end{align}
so the Cauchy--Schwartz and H\"{o}lder inequalities are effective in controlling products of $\eta$ with other entropy-bounded terms.

\item We move on to the thermal term. A direct computation shows that 
\begin{align}\label{relfmfM}
\mathcal{R}_{M} = \int d^n q \ \left( \rho_0 h\left(\frac{\rho-\rho_0}{\rho_0}\right) + \frac{n}{2}\rho k\left(\frac{\theta - \theta_0}{\theta_0}\right) + \rho \frac{\lvert u - u_0\rvert^2}{2\theta_0} \right),
\end{align}
where $h(z) = (1+z)\log(1+z) - z$ and $k(z) = z - \log(1+z)$. Both $h$ and $k$ are non-negative convex functions that are locally quadratic for small $z$, but subquadratic for $z \gg 1$. 
For infinitesimal fluctuations $\mathcal{R}_{M}$ behaves like an $L^2$-type norm, but for finite fluctuations, control on the relative entropy typically only implies control on the $L^1$-size of the fluctuations (e.g. by an argument in \cite{Cercignani88}, page 20). Nonetheless, by considering \emph{renormalised fluctuations} (\cite{SaintRaymond09book}, page 52), we can obtain a sharper result that is sufficient for our purposes. Define the \emph{renormalised thermal fluctuation} $\hat{\theta}$ by
\begin{align}\label{Trenom_decomp}
\hat{\theta} = \sqrt{\frac{\theta}{\theta_0}} - 1, \quad \text{so } \  \frac{\theta - \theta_0}{\theta_0} = 2\hat{\theta} + \hat{\theta}^2.
\end{align}
Then the elementary inequality $k(z)\geq (\sqrt{1+z}-1)^2$ implies that
\begin{align}
\int d^nq \ \rho \hat{\theta}^2 \leq \int d^n q \ \rho k\left(\frac{\theta - \theta_0}{\theta_0}\right) \leq \frac{2}{n}\mathcal{R}_{M},
\end{align}
so the renormalised thermal fluctuation $\hat{\theta}$ is $L^2$-controlled. While this only implies $L^1$-control on the true thermal fluctuation $(\theta-\theta_0)/\theta_0$, the decomposition (\ref{Trenom_decomp}) gives us a handle on the products of $(\theta-\theta_0)$ with other entropy-bounded terms, in particular $\eta$.
\end{enumerate}
Now we are ready to derive a bound for $\Delta H$. We decompose
\begin{align}
\Delta H = \frac{n}{2}\int d^nq \ \rho \theta \eta = \theta_0\left(  \frac{n}{2}\int d^nq \ \rho \eta + \frac{n}{2}\int d^nq \ \rho \frac{\theta - \theta_0}{\theta_0} \eta\right).
\end{align}
The first term in the brackets is bounded by $\mathcal{R}_{m}$ by (\ref{S_ineq1}). To bound the second term,
\begin{align}
\frac{n}{2}\int d^nq \ \rho \frac{\theta - \theta_0}{\theta_0} \eta = \ & \frac{n}{2}\int d^nq \ \rho \left(2\hat{\theta} + \hat{\theta}^2\right) \eta\nonumber \\
= \ & n\int d^nq \ \rho \hat{\theta}\eta + \frac{n}{2}\int d^nq \ \rho \hat{\theta}^2 \eta, \nonumber \\
\leq \ & n\left(\int d^nq \rho\hat{\theta}^2 \right)^{1/2}\left(\int d^nq \rho \eta^2 \right)^{1/2} +  \frac{n}{2}\int d^nq \ \rho \hat{\theta}^2, \nonumber \\
\leq \ & n\left(\frac{2}{n}\mathcal{R}_{M}\right)^{1/2}\left(\frac{2}{n}\mathcal{R}_{m}\right)^{1/2} + \mathcal{R}_{M} \nonumber \\
\leq \ & \mathcal{R}_{M} + \mathcal{R}_{m} + \mathcal{R}_{M}.
\end{align}
Thus $\Delta H \leq \theta_0(2\mathcal{R}_{M} + 2\mathcal{R}_{m}) = 2\theta_0\mathcal{R}_{in}$, which is precisely the claimed inequality (\ref{dH_inequality0}).
\end{proof}
The inequality (\ref{dH_inequality0}) is effective as $\mathcal{R}_{in} \rightarrow 0$. In this limit, the Mach number of the gas is necessarily small, since
\begin{align}
\int d^nq \ \rho \frac{\lvert u - u_0\rvert^2}{2\theta_0} \leq \mathcal{R}_{in},
\end{align}
by (\ref{relfmfM}).

The inequality (\ref{dH_inequality0}) holds regardless of how $f$ evolves. However, if $f$ does not evolve according to the collisionless Boltzmann equation, the right-hand side of (\ref{dH_inequality0}) is not necessarily a constant of motion.
Nonetheless, if $f$ evolves according to the \emph{Vlasov--Poisson equation} (see section \textbf{\ref{subsec:LPforKT}}), a statement similar to Proposition \ref{boundsprop} holds.
We show this by finding constants of motion of the Vlasov--Poisson equation that bound $\theta_0$ and $\mathcal{R}_{in}$ from above. 

Define the modified temperature $\Phi_0$ by
\begin{align}\label{Vlasov_Phidef}
\rho_0 \left(\lvert u_0\rvert^2 + n\Phi_0 \right) \mathrm{vol}\left(\mathbb{T}^n\right) & = \int d^n q d^n p \ \lvert p\rvert^2 f(q,p) + {e}^2\int d^nq d^nq^{\prime} \rho(q)\rho(q^{\prime})G(q,q^{\prime}).
\end{align}
The right-hand side is $2H_{VP}[f]$ (\ref{HVlasov}), which is a constant of motion of the Vlasov--Poisson equation. Since $\rho_0$ and $u_0$ are constants of motion of the Vlasov--Poisson equation as well, so is $\Phi_0$. The integral operator with kernel $G(q,q^{\prime})$ is positive semidefinite, so comparing (\ref{theta0def}) and (\ref{Vlasov_Phidef}) gives $\theta_0\leq\Phi_0$.

We then seek a constant of motion $\mathcal{S}_{in}$ that bounds $\mathcal{R}_{in}$ from above. The explicit expression for $\mathcal{R}_{in}$ (\ref{Hindef}) is
\begin{align}\label{Hinexpr}
\mathcal{R}_{in} = \int d^nq d^np \ f\log f +  \rho_0\mathrm{vol}\left(\mathbb{T}^n\right)\left(-\log\rho_0 + \frac{n}{2}\log\theta_0 + \frac{n}{2}\log\left(2\pi\mathrm{e}\right)\right).
\end{align}
Since $\log \theta_0 \leq \log \Phi_0$, choosing
\begin{align}\label{Sinexpre}
\mathcal{S}_{in} = \int d^nq d^np \ f\log f + \rho_0\mathrm{vol}\left(\mathbb{T}^n\right)\left(-\log\rho_0 + \frac{n}{2}\log\Phi_0 + \frac{n}{2}\log\left(2\pi\mathrm{e}\right)\right)
\end{align}
gives an upper bound for $\mathcal{R}_{in}$ by constants of motion.
\begin{corollary}
Let $\Delta H[f]$ be as in (\ref{cannotreduce}). Suppose that $f$ evolves according to the Vlasov--Poisson equation. Then the following inequality holds:
\begin{align}\label{dH_inequality0VP}
\Delta H[f] \leq 2\Phi_0 \mathcal{S}_{in},
\end{align}
where the right-hand side consists only of constants of motion.
\end{corollary}
A similar statement for the \emph{Jeans equation} of a self-gravitating system does not hold however, since the potential energy in the self-gravitating Hamiltonian $H_{SG}[f]$ (\ref{HJeans}) has the opposite sign. While the inequality (\ref{dH_inequality0}) still holds, we cannot find an upper bound for $\theta_0$ in terms of constants of motion by the same method due to the sign difference.

\subsection{Approximate reduction for isotropic near local Maxwellian distribution functions}\label{subsec:isotropic}

In this section we develop an approximation of $\Delta H$ (\ref{deltaHdef}) for nearly local Maxwellian and \emph{isotropic} distribution functions in terms of the generalised entropy variables constructed in section \textbf{\ref{subsec:JA}}.
Let $f = f_m (1 + \epsilon h)$, where $f_m$ is the local Maxwellian associated to $f$ (\ref{localMdef}), and $\epsilon$ be a small parameter. In what follows, we will consider a formal perturbation expansion in $\epsilon$.
The relative entropy density (\ref{relentloc}) of $f$ against $f_m$ is then
\begin{align}\label{rffmepsilon}
r[f\vert f_m] = \frac{\epsilon^2}{2}\int d^n p \ f_m h^2 + O(\epsilon^3).
\end{align}
Substituting (\ref{rffmepsilon}) in (\ref{theta2}) gives $\theta - T = O(\epsilon^2)$, whence
\begin{align}\label{DeltaHhL2}
\Delta H = \frac{\epsilon^2}{2}\int d^n q d^n p \ \theta f_m h^2 + O(\epsilon^3) = \frac{\epsilon^2}{2}\int d^n q d^n p \ T f_m h^2 + O(\epsilon^3),
\end{align}
so $\Delta H$ behaves like a weighted $L^2$-norm on the non-Maxwellian deviation $h$ in the near local Maxwellian limit $\epsilon \rightarrow 0$. An approximate reduction becomes possible if we can express $\int d^np f_m h^2$ in terms of hydrodynamic variables to leading order in $\epsilon$. We show that this can be achieved when $h$ is \emph{isotropic} in $p$, namely that
\begin{align}
h(q,p) = h(q, \chi), \qquad \text{where } \chi = \frac{\lvert p - u \rvert^2}{2\theta}.
\end{align}
The assumption that $f$ is isotropic is known as the \emph{Eddington approximation} in the kinetic theory of radiative transfer \cite{AndersonSpiegel72,Krook55,UnnoSpiegel66}. For any isotropic function $g$, we have
\begin{align}
\int d^n p \ f_m g = \rho \int_0^\infty d\chi \ F_n  g, \qquad \text{where } F_n(\chi) = \frac{1}{\Gamma(n/2)}\chi^{({n}/{2})-1} \exp(-\chi).
\end{align}
Let us also expand $h$ in terms of \emph{(generalised) Laguerre polynomials}, following the conventions in \cite{AndrewsAskeyRoy1999bookCh6}:
\begin{align}\label{hLagurrre}
h(q,\chi) = \sum_{b=2}^\infty \beta_b(q) L^{({n}/{2}) -1}_b(\chi),
\end{align}
where the $b=0$ and $b=1$ terms are dropped because $f$ is required to have the same $1$ and $\chi$ moments as $f_m$. The expression of $\Delta H$ in terms of the $\beta_b$ is
\begin{align}\label{DeltaHbetas}
\Delta H = \frac{\epsilon^2}{2}\int d^n q \rho T \sum_{b=2}^\infty \frac{\Gamma(b + n/2)}{\Gamma(b+1)\Gamma(n/2)} \beta_b^2 + O(\epsilon^3).
\end{align}
In the following we show that the $\beta_b$ can be expressed solely in terms of the density $\rho$ and the generalised entropy densities (\ref{Jdefcpts}) $s_1,s_2,\ldots s_b$ to leading order in $\epsilon$, namely that
\begin{align}\label{betasolnform}
\beta_b = \frac{1}{\epsilon}\tilde{\beta}_b(\rho,s_1,s_2,\ldots,s_b) + O(\epsilon),
\end{align}
which permits $\Delta H$ (\ref{DeltaHbetas}) to be approximated by hydrodynamic variables with $O(\epsilon^3)$ error. The $\tilde{\beta}_b$ are $O(\epsilon)$ quantities formed by differences of nominally $O(1)$ quantities and have no explicit $\epsilon$-dependence, as we shall see.

For notational simplicity, we write $\eta_a = s_a/\rho$ for $a\geq 1$. From (\ref{Jdefcpts}) we have
\begin{align}\label{etaaeqn}
\eta_a = & \frac{1}{\rho} \int d^n p \ f(\log f)^a = \int d^n p \ \frac{f_m}{\rho}(1 + \epsilon h) \left(\log f_m + \log(1+\epsilon h) \right)^a, \nonumber \\
= & \int_0^\infty d\chi \ F_n\left(\log\left(\frac{\rho}{(2\pi\theta)^{n/2}}\right) - \chi \right)^a \nonumber \\
\qquad & + \epsilon  \int_0^\infty d\chi \ F_n h \left(\log\left(\frac{\rho}{(2\pi\theta)^{n/2}}\right) - \chi \right)^{a-1}\left(\log\left(\frac{\rho}{(2\pi\theta)^{n/2}}\right) +a - \chi \right) + O(\epsilon^2).
\end{align}
The $O(1)$ term in the expression (\ref{etaaeqn}) for $\eta_a$ is the value of $\eta_a$ when $f$ is a local Maxwellian, which we denote as $\bar{\eta}_a$. More importantly, the $O(\epsilon)$ term in (\ref{etaaeqn}) is a \emph{linear} function of $\beta_2,\ldots \beta_a$ due to the properties of the Laguerre polynomials. From the $a=1$ case of (\ref{etaaeqn}) we find
\begin{align}\label{eta1eqn}
\log\left(\frac{\rho}{(2\pi\theta)^{n/2}}\right) = \eta_1 + \frac{n}{2} - \frac{r[f\vert f_m]}{\rho} = \eta_1 + \frac{n}{2}+ O(\epsilon^2),
\end{align}
so we can rewrite (\ref{etaaeqn}) for $a \geq 2$ as a linear system of equations
\begin{align}\label{etaMeqn}
\eta_a - \bar{\eta}_a(\eta_1) = \epsilon \sum_{b=2}^\infty M(\eta_1)_{ab} \beta_b + O(\epsilon^2),
\end{align}
where $M_{ab}$ is a lower triangular matrix with nonzero diagonal entries. The left-hand side $\eta_a - \bar{\eta}_a(\eta_1)$ of (\ref{etaMeqn}) is an $O(\epsilon)$ quantity that is the difference of two nominally $O(1)$ quantities. The first few terms are
\begin{align}
\begin{pmatrix}
\bar{\eta}_2 \\
\bar{\eta}_3 \\
\bar{\eta}_4 
\end{pmatrix}
 & =
\begin{pmatrix}
\eta_1^2 + \frac{1}{2}n \\
\eta_1^3 + \frac{3}{2}n\eta_1 - n \\
\eta_1^4 + 3n\eta_1^2 - 4n\eta_1 + \frac{3}{4}n^2 + 3n 
\end{pmatrix}
,\nonumber \\ 
M(\eta_1) & = 
\begin{pmatrix}
\frac{1}{4}n(n + 2) &  0 & 0 \\
\frac{3}{4}n(n + 2)(\eta_1 - 1) & \frac{1}{8}n(n^2 + 6n + 8) &  0 \\
\frac{3}{4}n(n + 2)(2\eta_1^2 - 4\eta_1 + n + 4) & \frac{1}{2}n(n^2 + 6n + 8)(\eta_1 - 2) & \frac{1}{16}n(n^3 + 12n^2 + 44n + 48)
\end{pmatrix}
.
\end{align}
The solution to (\ref{etaMeqn}) is
\begin{align}\label{etaMsoln}
\beta_b = \frac{1}{\epsilon} \sum_{c=2}^\infty (M^{-1})(\eta_1)_{bc} \left( \eta_c - \bar{\eta}_c(\eta_1) \right) + O(\epsilon),
\end{align}
where $M^{-1}_{bc}$, the inverse matrix of $M_{ab}$, is again lower triangular with nonzero diagonal entries. We have thus arrived at (\ref{betasolnform}). The expressions for the first few $\bar{\beta}_b$ are
\begin{align}
\tilde{\beta}_2 & = \frac{2\eta_2 - 2\eta_1^2 - n}{n(n + 2)/2}, \\
\tilde{\beta}_3 & =  \frac{4\eta_1^3 - 6\eta_1^2 - 6\eta_2\eta_1 + 6\eta_2 + 2\eta_3 - n}{n(n^2 + 6n + 8)/4}, \\
\tilde{\beta}_4 & = \frac{48\eta_2 + 32\eta_3 + 4\eta_4 - 4n - 96\eta_1\eta_2 - 16\eta_1\eta_3 - 12n\eta_2 + 24\eta_1^2\eta_2 + 12n\eta_1^2 - 48\eta_1^2 + 64\eta_1^3 - 12\eta_1^4 + 3n^2}{n(n^3 + 12n^2 + 44n + 48)/4} .
\end{align}
Returning to the problem of approximating $\Delta H$, consider for $A = 2,3,\ldots$ the quantities
\begin{align}
\Delta H_A [f] = \frac{1}{2}\int d^n q \rho[f] T(\rho[f],s_1[f]) \sum_{b=2}^A \frac{\Gamma(b + n/2)}{\Gamma(b+1)\Gamma(n/2)} \bar{\beta}_b(\rho,s_1,s_2,\ldots,s_b)^2,
\end{align}
and let $\Delta H_\infty = \lim_{A\rightarrow \infty} \Delta H_A$. We know from (\ref{DeltaHbetas}) that $\Delta H - \Delta H_\infty = O(\epsilon^3)$. Each $\Delta H_A$ factors through the Poisson map $\mathcal{J}_A: C^\infty(T^*M)_+ \rightarrow \mathfrak{s}^*_A$ (\ref{Jdef}). The infinite sum $\Delta H_\infty$ does not factor through $\mathcal{J}_A$ for any finite $A$, but it factors through the larger Poisson map $\mathsf{J}:C^\infty(T^*M)_+ \rightarrow \mathrm{Vect}(M)^* \times P\Omega^n(M)$ (\ref{Jbigdef}), since the Tsallis entropy density $\rho_\xi[f]$ (\ref{Jepsdefcpts}) is the generating function for all the generalised entropy densities $s_a[f]$.

We remark that the isotropic assumption of $h$ is itself a truncation of the full tensor Hermite polynomial expansion of $h$ by discarding the non-isotropic moments \cite{AndersonSpiegel72,Krook55,UnnoSpiegel66}. While the relationship (\ref{DeltaHhL2}) between $\Delta H$ and $h$ holds in the near local Maxwellian regime regardless of whether $h$ is isotropic, we cannot obtain the contribution from the non-isotropic parts of $h$ to $\Delta H$ using the generalised entropies.
The non-isotropic part of $h$ does not enter the $O(\epsilon)$ part of (\ref{etaaeqn}) because $\log f_m$ is an isotropic function.
If one attempts to continue the expansion $h = h^{(0)} + \epsilon h^{(1)} + \ldots$, the $O(\epsilon^2)$ part of (\ref{etaaeqn}) gives an unclosed relationship between the isotropic part of $h^{(1)}$ and the non-isotropic part of $h^{(0)}$, both of which are unknown.

The use of $\Delta H_A$ instead of the full series $\Delta H_\infty$ corresponds to a further truncation of the Laguerre polynomial expansion of $h$ (\ref{hLagurrre}) at order $A$
\begin{align}\label{htruncatedLaguerre}
h(q,\chi) = \sum_{b=2}^A \beta_b(q) L^{({n}/{2}) -1}_b(\chi),
\end{align}
so hydrodynamic models formed with the Hamiltonian $H_{fluids}[m,\rho,s_1] + \Delta H_A[\rho,s_1,\ldots,s_A]$ on $\mathfrak{s}^*_A$ can be considered as a finite isotropic moment truncation for the non-Maxwellian deviation $h$ in the near local Maxwellian limit $\epsilon \rightarrow 0$. For example, for $A = 2$, such a Hamiltonian on $\mathfrak{s}^*_2$ is
\begin{align}
H[m,\rho,s_1,s_2] = \int d^nq \ \left(\frac{\lvert m\rvert^2}{2\rho} + \frac{n}{2}\rho T(\rho,s_1)\left(1+\frac{(2s_2 - 2s_1^2 - n\rho)^2}{2n^2(n + 2)\rho^2} \right)\right),
\end{align}
where $T$ is as in (\ref{EOS_idealgas}).

\section{Conclusion}

Kinetic theory and fluid dynamics are both \emph{noncanonical Hamiltonian systems} in the absence of dissipative effects such as collisions, viscous friction and thermal diffusion. The time evolution in each theory is determined by a Poisson bracket and a Hamiltonian functional. This manuscript aims to study the relationship between the two Hamiltonian structures of the respective theories. In particular, we investigated the possibility of a \emph{Hamiltonian reduction through a Poisson map} (section \textbf{\ref{subsec:myreduction}}): given a more primitive noncanonical Hamiltonian system $(\mathfrak{M}_1,\{\cdot,\cdot\}_1,\tilde{H})$ and a reduced noncanonical Hamiltonian system $(\mathfrak{M}_2,\{\cdot,\cdot\}_2,H)$, we would like to relate the two with a map $\mathcal{J}:\mathfrak{M}_1 \rightarrow \mathfrak{M}_2$ that respects the Poisson brackets on the two spaces (i.e. a Poisson map), and such that $\tilde{H} = H \circ \mathcal{J}$.

It is known that there exists a Poisson map $\mathcal{J}_0$ from the space of $1$-particle distribution functions in kinetic theory to the space of hydrodynamic variables excluding the entropy density \cite{Guillemin80, Marsden83}. We generalised this to a Poisson map $\mathcal{J}_1$  (section \textbf{\ref{subsec:JA}}), which takes the $p$-integral of the Boltzmann entropy $f\log f$ to the hydrodynamic entropy density. The map $\mathcal{J}_1$ is the second member of a family of Poisson maps $\mathcal{J}_A$, which include generalised entropy densities $s_a$, given as the $p$-integrals of $f(\log f)^a$ for $a = 0,1,\ldots A$, as additional hydrodynamic variables (section \textbf{\ref{subsec:JA}}). The maps $\mathcal{J}_A$ are in turn obtained from the Taylor expansion of a Poisson map $\mathsf{J}$ including a formal parameter $\epsilon$, which maps the $p$-integral of the Tsallis entropy \cite{Tsallis88} $f^{1+\epsilon}$ to a hydrodynamic variable (section \textbf{\ref{subsec:Jeps}}).

However, the kinetic theory Hamiltonian $H_{KT}[f]$ (\ref{HKT_idealgas}) does not factor through this Poisson map, as it depends on the $|p|^2$-moment of the distribution function $f$, which cannot be expressed in terms of the image $\mathcal{J}_1[f]$. This difficulty is related to the \emph{moment closure problem} in kinetic theory, where kinetic theory cannot be described exactly by a finite number of $p$-integrals involving $f$. We can therefore only hope for \emph{approximate closure}, where we make some uncontrolled approximations or impose some constraints to arrive at a reduced description. In our case, we replaced the kinetic theory temperature $\theta[f]$ (\ref{theta0}), which is defined in terms of the local root-mean-squared velocity deviation from the mean, with the thermodynamic temperature $T(\rho,s)$ (\ref{EOS_idealgas}), which is defined in terms of the local ideal gas equation of state.
This is equivalent to ignoring the contribution from the relative entropy density $r[f|f_m]$ (\ref{relentloc}) of $f$ against its local Maxwellian $f_m$ (\ref{localMdef}) in the Hamiltonian functional $H_{KT}[f]$. After making this uncontrolled approximation, the resulting effective Hamiltonian is precisely the pullback of the fluid Hamiltonian $H_{fluids}[m,\rho,s]$ (\ref{Hfluids_idealgas}) through $\mathcal{J}_1$ (section \textbf{\ref{sec:Hamiltonian}}). Unlike perturbation methods such as the Chapman--Enskog expansion \cite{ChapmanCowling70book, Cercignani90book}, our approach has no bearing on the physical validity of the reduced equations obtained; however, our approach reveals that the passage from kinetic theory to ideal fluid dynamics can be done in a manifestly Hamiltonian manner (section \textbf{\ref{subsec:Hdiscuss}}). We also obtained an analogous relationship between the Vlasov--Poisson equation (\ref{Vlasov Poisson}) and the Euler--Poisson equations \cite{DietzSandor99, ZakharovKuznetsov97} with non-constant entropy.
We conclude with an investigation on the ignored term $\Delta H$ (\ref{deltaHdef}). We obtained a bound on $\Delta H$ with constants of motion in the near global Maxwellian regime (section \textbf{\ref{subsec:globalM}}), and obtained an approximation for $\Delta H$ in terms of the generalised entropy densities $s_a$ for isotropic and near local Maxwellian distribution functions (section \textbf{\ref{subsec:isotropic}}).

We remark that there are other Hamiltonian approaches to the moment closure problem. The first is to give up finiteness and allow an infinite number of $p$-integrals as hydrodynamic variables \cite{Gibbons81, GibbonsHolmTronci08}, in which case the map from kinetic distribution functions to hydrodynamic variables is still a Poisson map. More recently, it has been popular in plasma physics to construct a new but related Poisson bracket on the space of hydrodynamic variables through imposing a suitable chosen closure relation \cite{Tassi15, Tassi16, Tassi17, PerinChandreMorrisonTassi15}. In these applications, the map from distribution functions to hydrodynamic variables is no longer a Poisson map, but the kinetic theory Hamiltonian is the pullback of the fluid Hamiltonian. We give a possible mathematical interpretation of these constructions in section \textbf{\ref{subsec:notmyreduction}}. The approach taken in this manuscript is to simply make an uncontrolled approximation in the Hamiltonian functional, which is a very common technique in theoretical physics and applied mathematics in deriving reduced Hamiltonian models \cite{MilesSalmon85, Salmon85, Whitham67,Whitham70,Whitham74book, MalomedWeinstein96, Kevrekidis09book, KhesinMisiolekModin20,KhesinMisiolekModin19, MarsdenMorrisonWeinstein84} (see section \textbf{\ref{subsec:Hdiscuss}} for examples and more discussion).

\section*{Acknowledgements}\label{sec:ack}

The author is grateful to Paul Dellar for his many helpful comments and suggestions. The author thanks the reviewers for suggesting the waterbag distribution functions in section \textbf{\ref{subsec:notmyreduction}} and for raising a discussion on momentum maps in section \textbf{\ref{subsec:JA}}. This work was supported by the Mathematical Institute, University of Oxford, which played no other role in the research, or in the preparation and submission of the manuscript.

\begin{appendices}

\section{Proofs of Propositions \ref{Jepsprop} and \ref{Jbigrop} }\label{App:proofs}

In this appendix we present the proofs to Propositions \ref{Jepsprop} and \ref{Jbigrop}. The core idea is to perform a computation similar to that in the proof of Proposition \ref{JPoissonprop}.

\subsection*{Proof of Proposition \ref{Jepsprop}}

\begin{proof}
%The proof is very similar to that of Proposition \ref{JPoissonprop}. 
We show by direct computation that, if $\tilde{F} = F \circ \mathsf{J}_{\xi}$ and $\tilde{G} = G \circ \mathsf{J}_{\xi}$, then
\begin{align}
\left\{\tilde{F},\tilde{G} \right\}_{KT} = \left\{{F},{G} \right\}_{\mathfrak{s}_0} \circ \mathsf{J}_{\xi} .
\end{align}
By the functional chain rule we have
\begin{align}\label{chainrule_epsilon}
\frac{\delta \tilde{F}}{\delta f} = \left\langle \frac{\delta F}{\delta m}, p \right\rangle + (1+\xi) f^\xi \frac{\delta F}{\delta \rho_{\xi}}.
\end{align}
Again, $\delta F/\delta m$ and $\delta F/\delta \rho_\xi$ have no $p$-dependence. The case $\xi = 0$ corresponds to the previously known Poisson map in \cite{Guillemin80, Marsden83}, which in our notation is denoted by $\mathsf{J}_{0}=\mathcal{J}_0$. Focusing on $\xi \neq 0$, we have
\begin{align}\label{Jeps0thline}
\left\{ \tilde{F}, \tilde{G} \right\}_{KT} = & \int d^n q d^n p \ f \left\{ \frac{\delta \tilde{F}}{\delta f}, \frac{\delta \tilde{G}}{\delta f} \right\} \nonumber \\
& \qquad + f \left\{(1+\xi) f^\xi \frac{\delta F}{\delta \rho_{\xi}}, \left\langle \frac{\delta G}{\delta m}, p \right\rangle  \right\} + f \left\{\left\langle \frac{\delta F}{\delta m}, p \right\rangle, (1+\xi)f^\xi\frac{\delta G}{\delta \rho_{\xi}}  \right\} \nonumber \\
& \qquad + f \left\{(1+\xi) f^\xi \frac{\delta F}{\delta \rho_{\xi}}, (1+\xi)f^\xi\frac{\delta G}{\delta \rho_{\xi}}  \right\}.
\end{align}
The term on the first line is $-\left\langle m[f], [\delta F/\delta m, \delta G/\delta m]\right\rangle_{\mathrm{Vect}(M)}$ by the computation in Proposition \ref{JPoissonprop}. To compute the terms on the second line, consider
\begin{flalign}\label{Jeps2ndline}
& \int d^nq d^n p \ f \left\{(1+\xi) f^\xi \frac{\delta F}{\delta \rho_{\xi}}, \left\langle \frac{\delta G}{\delta m}, p \right\rangle  \right\} \nonumber &\\
= & \int d^nq d^n p \ (1+\xi) f^{1+\xi} \left\{\frac{\delta F}{\delta \rho_{\xi}}, \left\langle \frac{\delta G}{\delta m}, p \right\rangle  \right\} + (1+\xi) f\frac{\delta F}{\delta \rho_{\xi}} \left\{f^\xi, \left\langle \frac{\delta G}{\delta m}, p \right\rangle  \right\} \nonumber &\\
= & \int d^nq d^n p \ (1+\xi) f^{1+\xi} \left\{\frac{\delta F}{\delta \rho_{\xi}}, \left\langle \frac{\delta G}{\delta m}, p \right\rangle  \right\} + \xi(1+\xi) f^\xi \frac{\delta F}{\delta \rho_{\xi}} \left\{f, \left\langle \frac{\delta G}{\delta m}, p \right\rangle  \right\} \nonumber &\\
= & \int d^nq d^n p \ (1+\xi) f^{1+\xi} \left\{\frac{\delta F}{\delta \rho_{\xi}}, \left\langle \frac{\delta G}{\delta m}, p \right\rangle  \right\} + \xi \frac{\delta F}{\delta \rho_{\xi}} \left\{f^{1+\xi}, \left\langle \frac{\delta G}{\delta m}, p \right\rangle  \right\} \nonumber &\\
= & \int d^nq d^n p \ (1+\xi) f^{1+\xi} \left\{\frac{\delta F}{\delta \rho_{\xi}}, \left\langle \frac{\delta G}{\delta m}, p \right\rangle  \right\} - \xi f^{1+\xi}  \left\{\frac{\delta F}{\delta \rho_{\xi}}, \left\langle \frac{\delta G}{\delta m}, p \right\rangle  \right\} \nonumber &\\
= &  \int d^nq d^n p \ f^{1+\xi} \left\{\frac{\delta F}{\delta \rho_{\xi}}, \left\langle \frac{\delta G}{\delta m}, p \right\rangle  \right\} = \left\langle \rho_{\xi}[f], \mathcal{L}_{\frac{\delta G}{\delta m}} \frac{\delta F}{\delta \rho_{\xi}} \right\rangle_{C^\infty(M)}.&
\end{flalign} 
The term on the third line of (\ref{Jeps0thline}) vanishes, because
\begin{flalign}\label{Jeps3rdline}
& \int d^n q d^n p \ (1+\xi)^2 f \left\{ f^\xi \frac{\delta F}{\delta \rho_{\xi}}, f^\xi\frac{\delta G}{\delta \rho_{\xi}}  \right\} &\nonumber \\
= &  \int d^n q d^n p \ (1+\xi)^2 f \bigg(
f^{2\xi}\left\{\frac{\delta F}{\delta \rho_{\xi}},\frac{\delta G}{\delta \rho_{\xi}} \right\}
+ f^\xi \frac{\delta F}{\delta \rho_{\xi}}\left\{f^\xi,\frac{\delta G}{\delta \rho_{\xi}} \right\}   \nonumber \\
& \qquad \qquad \qquad \qquad \qquad 
+ f^\xi \frac{\delta G}{\delta \rho_{\xi}}\left\{\frac{\delta F}{\delta \rho_{\xi}} ,f^\xi\right\}  
+ \frac{\delta F}{\delta \rho_{\xi}}\frac{\delta G}{\delta \rho_{\xi}}\left\{f^\xi,f^\xi \right\}
\bigg) &\nonumber \\
= & \int d^n q d^n p \ (1+\xi)^2 f^{1+\xi} \left(
\frac{\delta F}{\delta \rho_{\xi}}\left\{ f^\xi,\frac{\delta G}{\delta \rho_{\xi}} \right\}
+ \frac{\delta G}{\delta \rho_{\xi}}\left\{\frac{\delta F}{\delta \rho_{\xi}} ,f^\xi\right\} 
\right) &\nonumber \\
= & \int d^n q d^n p \ \xi(1+\xi)^2 f^{2\xi}\left(
\frac{\delta F}{\delta \rho_{\xi}}\left\{ f,\frac{\delta G}{\delta \rho_{\xi}} \right\}
+ \frac{\delta G}{\delta \rho_{\xi}}\left\{\frac{\delta F}{\delta \rho_{\xi}} ,f \right\} 
\right) &\nonumber \\
= & \int d^n q d^n p \ \frac{\xi(1+\xi)^2}{1+2\xi}\left(
\frac{\delta F}{\delta \rho_{\xi}}\left\{ f^{1+2\xi},\frac{\delta G}{\delta \rho_{\xi}} \right\}
+ \frac{\delta G}{\delta \rho_{\xi}}\left\{\frac{\delta F}{\delta \rho_{\xi}} ,f^{1+ 2 \xi} \right\} 
\right) &\nonumber \\
= & \int d^n q d^n p \ \frac{-2\xi(1+\xi)^2}{1+2\xi}f^{1+2\xi}
\left\{ \frac{\delta F}{\delta \rho_{\xi}},\frac{\delta G}{\delta \rho_{\xi}} \right\} = 0.&
\end{flalign}
We have used the identity (\ref{ydiff}) extensively to rearrange terms of the form $f^\lambda\{f^\mu,\ldots\}$. Putting everything together, we have
\begin{align}
\left\{\tilde{F},\tilde{G} \right\}_{KT} = -\left\langle m[f], \left[\frac{\delta F}{\delta m}, \frac{\delta G}{\delta m}\right]\right\rangle_{\mathrm{Vect}(M)} - \left\langle \rho_{\xi}[f], \mathcal{L}_{\frac{\delta F}{\delta m}} \frac{\delta G}{\delta \rho_{\xi}} - \mathcal{L}_{\frac{\delta G}{\delta m}} \frac{\delta F}{\delta \rho_{\xi}} \right\rangle_{C^\infty(M)},
\end{align}
which is precisely the ($-$)-Lie--Poisson bracket on $\mathfrak{s}^*_0$ evaluated at the image of $\mathsf{J}_{\xi}$.
\end{proof}

\subsection*{Proof of proposition \ref{Jbigrop}}

\begin{proof}
We repeat the argument used in Proposition \ref{Jepsprop}. We replace (\ref{chainrule_epsilon}) with
\begin{align}
\frac{\delta \tilde{F}}{\delta f} = \left\langle \frac{\delta F}{\delta m}, p \right\rangle + \int_I d\xi \ (1+\xi) f^\xi \frac{\delta F}{\delta \rho_{\xi}},
\end{align}
and compute $\{ \tilde{F}, \tilde{G} \}_{KT}$. Since integration in $\xi$ can be pulled out of the canonical Poisson bracket, we can separate $\{ \tilde{F}, \tilde{G} \}_{KT}$ into three groups of terms, each with $0,1$ or $2$ $\xi$-integrals respectively. The first group of terms with no $\xi$-dependence is just $-\langle m[f], [\delta F/\delta m, \delta G/\delta m]\rangle$. For the second group of terms with a single $\xi$-integral, repeating the computation in (\ref{Jeps2ndline}) shows that
\begin{flalign}
& \int_I d\xi \int d^nqd^np \ f \left\{(1+\xi) f^\xi \frac{\delta F}{\delta \rho_{\xi}}, \left\langle \frac{\delta G}{\delta m}, p \right\rangle  \right\} \nonumber &\\
= & \int_I d\xi \int d^nq d^n p \ f^{1+\xi} \left\{\frac{\delta F}{\delta \rho_{\xi}}, \left\langle \frac{\delta G}{\delta m}, p \right\rangle  \right\}, \nonumber &\\
= & \int_I d\xi \ \left\langle \rho_{\xi}[f], \mathcal{L}_{\frac{\delta G}{\delta m}} \frac{\delta F}{\delta \rho_{\xi}} \right\rangle_{C^\infty(M)}.&
\end{flalign}
The third group of terms can be written as
\begin{align}\label{Jbig3rd}
\int_{I\times I} d\xi d\xi^\prime \int d^n q d^n p \ (1+\xi)(1+\xi^\prime) f \left\{ f^\xi \frac{\delta F}{\delta \rho_{\xi}}, f^{\xi\prime}\frac{\delta G}{\delta \rho_{\xi^\prime}}  \right\}.
\end{align}
A computation entirely analogous to (\ref{Jeps3rdline}) shows that the term (\ref{Jbig3rd}) vanishes, because ${\delta F}/{\delta \rho_{\xi}}$ and ${\delta G}/{\delta \rho_{\xi^\prime}}$ are functions of $q$ only for fixed $\xi, \xi^\prime$. Putting everything together, we have
\begin{align}
\left\{\tilde{F},\tilde{G} \right\}_{KT} = &-\left\langle m[f], \left[\frac{\delta F}{\delta m}, \frac{\delta G}{\delta m}\right]\right\rangle_{\mathrm{Vect}(M)} \nonumber \\
 & \qquad - \int_I d\xi \ \left\langle \rho_{\xi}[f], \mathcal{L}_{\frac{\delta F}{\delta m}} \frac{\delta G}{\delta \rho_{\xi}} - \mathcal{L}_{\frac{\delta G}{\delta m}} \frac{\delta F}{\delta \rho_{\xi}} \right\rangle_{C^\infty(M)},
\end{align}
which proves that $\mathsf{J}$ is indeed a Poisson map.
\end{proof}

\end{appendices}

\bibliography{references}{}
\bibliographystyle{abbrv}
\end{document}